\documentclass[12pt,letterpaper]{article}

\usepackage{lmodern}
\usepackage[T1]{fontenc}
\usepackage[utf8]{inputenc}

\usepackage{epigraph} 

\usepackage[english]{babel}

\usepackage[letterpaper,margin=1 in]{geometry}
\usepackage{setspace}
\usepackage[protrusion=true,expansion=true,final]{microtype}

\usepackage{amsmath,amsthm,mathtools}
\usepackage{amsfonts, amssymb}
\usepackage{bm}    % bold math
\usepackage{bbm}   % blackboard-bold 1

\newtheorem{lemma}{Lemma}
\newtheorem{proposition}{Proposition}
\newtheorem{assumption}{Assumption}
\newtheorem{definition}{Definition}

\usepackage{graphicx}
\usepackage{subcaption}               % modern sub-figures
\usepackage[font={small,stretch=1.1},labelfont=bf]{caption}
\usepackage{booktabs}                 % professional tables
\usepackage[flushleft]{threeparttable}
\usepackage{siunitx}                  % decimal-aligned numbers
\usepackage{dcolumn}
\newcolumntype{d}[1]{D{.}{.}{#1}}     % e.g. d{2} for two-decimals
\usepackage{tikz}
\usetikzlibrary{positioning, arrows.meta}
\usepackage{booktabs}       % For \toprule, \midrule, \bottomrule
\usepackage{threeparttable} % For the threeparttable environment and \tnote
\usepackage{tabularx}
\usepackage{algorithm}
\usepackage[noend]{algorithmic}       % omit “end” keywords
\usepackage{listings}                 % source-code blocks
\usepackage{xcolor}
\definecolor{linkcol}{HTML}{0066CC}                     % bright cobalt
\definecolor{lightgray}{gray}{0.85}

\usepackage{enumitem}

\usepackage[compact]{titlesec}
\titlespacing{\section}{0pt}{2ex}{1ex}
\titlespacing{\subsection}{0pt}{1ex}{0ex}
\titlespacing{\subsubsection}{0pt}{0.5ex}{0ex}

\usepackage{natbib}

\usepackage{xcolor}
\definecolor{linkcol}{HTML}{db0404}
\usepackage[
    colorlinks      = true,
    linkcolor       = linkcol,
    citecolor       = linkcol,
    urlcolor        = linkcol,
    backref         = page          % show citing pages in bibliography
]{hyperref}

\usepackage[capitalize,nameinlink,noabbrev]{cleveref}
\crefname{assumption}{Assumption}{Assumptions}

\usepackage{fancyhdr}
\newcommand{\IR}{\text{IR}}
\newcommand{\IRSlack}{\text{IR-Slack}}

\begin{document}

\title{Automatic Renewal Trials and Rationally Inattentive Customer}
\author{Felicia Nguyen\thanks{Department of Marketing - Emory University. \href{mailto:pnguy38@emory.edu}{pnguy38@emory.edu} }\\
Emory University}
\date{\today}
\maketitle

\begin{abstract}
\small
\onehalfspacing
\noindent The ``free trial'' followed by automatic renewal is a prevalent business model in the digital economy. While traditional models view trials as a way for consumers to learn a product's value, we propose a complementary theory based on rational inattention \citep{Sims2003}. In our framework, consumers already know their valuation but face cognitive costs associated with remembering to cancel an unwanted subscription. We model this using a Shannon attention cost, where a consumer's attention decays over the trial period. This creates a fundamental trade-off for firms: longer trials increase ``inattentive revenue'' from consumers who forget to cancel, but they also reduce the initial attractiveness of the offer. This dynamic leads to an optimal trial length, even when consumers do not need time to learn the product's value. Our model predicts that optimal renewal prices and trial lengths are complements; longer trials are associated with higher post-trial prices. Analyzing consumer protection measures, such as click-to-cancel regulations, we find that making cancellation easier motivates firms to shorten trial periods. Furthermore, we demonstrate that introductory prices (paid trials) and trial lengths act as strategic substitutes. This framework provides a new lens for understanding subscription contracts and evaluating consumer protection policies in digital markets.
\end{abstract}


\newpage

%%%%%%%%%%%%%%%%%%%%%%%%%%%%%%%%%%%%%%%%%%%%%%%%%%%%%%%%%%%%%%%%%%%%%%
\section{Introduction}\label{sec:intro}
%%%%%%%%%%%%%%%%%%%%%%%%%%%%%%%%%%%%%%%%%%%%%%%%%%%%%%%%%%%%%%%%%%%%%%
\setlength{\epigraphwidth}{0.7\textwidth}
\epigraph{\textit{``The struggle of man against power is the struggle of memory against forgetting.''}}{Milan Kundera}

The proliferation of subscription-based services has become a defining feature of the modern economy. From streaming entertainment (e.g., Netflix, Spotify) and software as a service (Adobe Creative Cloud, Microsoft 365) to news media and e-commerce (Amazon Prime), or even meal preps (e.g., Home Chef, Factor), the ``free trial with automatic paid subscription renewal'' model is ubiquitous. A central approach in the customer acquisition and growth strategy toolkit for these services is an introductory trial offer, typically a free or low-cost period, after which the consumer is automatically enrolled in a recurring payment plan (commonly referred to as the ``negative option'' model by regulatory bodies). The subscription mode accounts for almost half the share of mobile apps revenue, despite the small share of subscription-based apps. Given the multi-billion-dollar scale of this market, understanding the strategic design of these contracts is highly valuable.

The dominant theoretical framework to examine product trial offers in existing literature follows the paradigm of consumer learning \citep{hoch1986consumer,iyengar2007model}. In this sense, consumers are uncertain about their valuation of a new product. A free trial serves as a mechanism to resolve this uncertainty, allowing consumers to learn their personal fit and make an informed purchase decision. The firm offers the trial to signal quality and persuade high-valuation consumers to adopt, balancing the cost of the trial against the future stream of revenue from converted subscribers. The optimal duration of such a trial is determined by the time required for this learning process to conclude.

While consumer learning is undeniably an important factor, it provides an incomplete picture of the subscription landscape. We often observe trial periods that appear excessively long, relative to any reasonable learning period. A user can likely assess a music streaming service's library and interface within a few days, yet 30-day, 60-day, or even longer trials are standard practice. Furthermore, a substantial body of evidence points to consumer inertia, forgetfulness, and procrastination as significant drivers of subscription retention. Surveys consistently reveal that a large number of consumers pay for subscriptions they no longer use, with forgetfulness cited as a primary reason (nearly half, according to a recent Forbes survey \citealt{orentas2024streaming}). This suggests that trial periods may serve a dual purpose: one of facilitating learning, and another, perhaps less benign, of exploiting consumer inattention.

\begin{figure}
    \centering
    \includegraphics[width= 0.8 \linewidth]{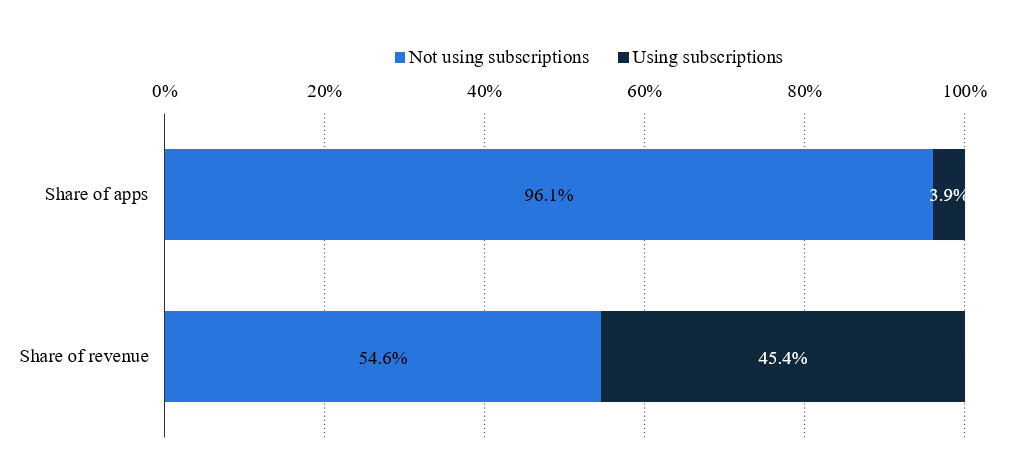}
    \caption{Distribution of subscription-based mobile apps and their revenue worldwide as of August 2024 (Source: Statista).}
    \label{fig:substats}
\end{figure}

We develop a formal theory of subscription contract design, in the case of auto-renewal, that isolates and explores this second channel. We propose a model where the primary strategic function of the trial length is to modulate consumer attention. To separate our mechanism from consumer learning, we assume that consumers know their valuation of the service early on\footnote{As discussed later in the paper, these two factors are likely complementary, and our assumption is purely to keep the model parsimonious.}. The central friction in our model is not uncertainty about value, but the cognitive cost of future actions. A consumer who signs up for a trial and finds that the service is not worth the renewal price must remember to perform a specific action, cancellation, by a specific deadline. This act of timely recall is not costless; it requires cognitive resources.

We formalize this friction using the rational inattention framework pioneered by \citet{Sims2003} and adapted for discrete choice by \citet{MatejkaMcKay2015}. In our model, a consumer who wishes to cancel can exert mental effort to increase the probability of successfully doing so. The cost of this effort is governed by their "attention sensitivity." Our key behavioral assumption, grounded in the psychology of prospective memory \citep{krishnan1999prospective, mcdaniel2000strategic}, is that this attention sensitivity decays with the length of the trial period. A task due tomorrow is more salient and easier to manage than a task due three months from now.

This intuitive assumption generates a novel and fundamental trade-off for the firm. On one hand, extending the trial period encourages attentional decay. This makes it cognitively more expensive for low-valuation consumers to ensure cancellation, leading to a higher rate of accidental renewals. This generates what we term \emph{inattentive revenue}, i.e., profit derived purely from consumer forgetfulness. This creates a strong incentive to lengthen the trial, a mechanism absent in standard learning models. On the other hand, the consumers in our model are sophisticated. They rationally anticipate their own cognitive limits and the associated costs of monitoring a distant deadline. A longer trial increases the risk of a costly mistake, lowering their ex-ante expected utility from signing up. We term this deterrent effect \emph{IR-Slack}. It acts as a powerful balancing force, incentivizing the firm to shorten the trial to increase initial adoption.

On the other hand, the consumers in our model are fully rational and sophisticated. They are aware of their own cognitive limitations and anticipate the challenges of managing a subscription with a distant deadline. They understand that a longer trial increases the expected cost of monitoring and the risk of a costly mistake. This potential cognitive burden lowers their expected utility from signing up for the subscription in the first place. We label this component the \textit{IR-Slack}, and this acts as a constraint, making the firm's offer appear less attractive and providing a strong incentive to shorten the trial period.

The optimal contract, consisting of the trial length ($T$) and renewal price ($P$), balances this tension, resulting in moderate-length trials we often observe in practice. The model thus explains why firms offer trials that are longer than necessary for learning but not exceedingly so. We then employ this framework to generate testable predictions about the structure of subscription contracts. By imposing standard assumptions on the distribution of consumer valuations, we find that trial length and renewal price are strategic complements. A firm that lengthens its trial period to enhance its inattentive revenue stream will also find it optimal to increase its renewal price to extract more value from this channel. This prediction diverges from the consumer learning framework, where the relationship between trial length and price is often unclear.

Our model also provides a lens for analyzing regulatory interventions. Concerns over ``dark patterns'' have spurred regulator proposals like ``click-to-cancel'' laws, recently introduced in California and proposed federally by the FTC \citep{Luguri2021}. We model such policies as an exogenous reduction in the cost of consumer attention. We show that this forces firms to retreat from the inattention margin, leading to shorter optimal trial lengths. The impact on pricing, however, is nuanced, depending crucially on the price elasticity of demand among willing subscribers. This highlights the necessity of theory-informed policy design to anticipate unintended consequences.

Finally, we extend our baseline model to the case of paid trials. We show that the introductory price and the trial length act as strategic substitutes. An introductory price functions as a \emph{screening device} to select for more committed, potentially higher CLV customers. This plan design will reduce the pool of ``low-valuation'' consumers who are the source of inattentive revenue, therefore diluting the incentive to use a long trial. This ``dilution principle'' illustrates how different contract features can substitute for each other in a firm's strategy to manage consumer behavior.

Our study contributes to three distinct strands of literature. First, we advance the rational inattention literature by endogenizing the cost of attention, wherein the firm’s contract design actively shapes the consumer's cognitive environment. Second, we contribute to behavioral industrial organization by providing a new micro-foundation for consumer inertia in a market populated by sophisticated consumers. Third, we offer a novel framework for business model strategy, positioning attention, rather than learning, as a central driver of subscription contract design.

The rest of the paper proceeds as follows. First, \Cref{sec:lit} provides a detailed review of the related literature. \Cref{sec:model} formally introduces the primitives of the model, as well as the description of the consumer's and firm's problems. In \Cref{sec:general} we present our main results that hold for a general class of information cost functions.  Next, \Cref{sec:shannon} specializes the model to the widely used Shannon entropy cost structure and imposes regularity conditions to derive sharper predictions. Additionally, in \Cref{sec:extension} we extend the model to the (less but still common) paid trials setting. In \Cref{sec:policy}, we synthesize the findings into a discussion of managerial and policy implications. Finally,  \Cref{sec:conclusion} concludes and outlines avenues for future research. Detailed proofs of our claims are provided in the accompanying Appendices.
%%%%%%%%%%%%%%%%%%%%%%%%%%%%%%%%%%%%%%%%%%%%%%%%%%%%%%%%%%%%%%%%%%%%%%
\section{Related Literature}\label{sec:lit}
%%%%%%%%%%%%%%%%%%%%%%%%%%%%%%%%%%%%%%%%%%%%%%%%%%%%%%%%%%%%%%%%%%%%%%

This research sits at the intersection of three influential streams of literature: the economic theory of rational inattention, behavioral industrial organization with a focus on contract design, and the literature on digital subscriptions and customer management. By synthesizing insights from each domain, we build a comprehensive theory of contract design based on attentional friction that bridges economic theory with managerially relevant practice.

\subsection{Rational Inattention and Bounded Rationality in Choice}

The foundation of our theoretical framework rests on the theory of rational inattention, introduced by \citet{Sims1998, Sims2003}. This theory provides a rigorous micro-foundation for information frictions, positing that agents are not perfectly informed, not because information is unavailable, but because processing such information is cognitively costly. Agents must therefore optimally allocate their finite attention. Our work builds on the tractable discrete-choice formulation of \citet{MatejkaMcKay2015}, where Shannon entropy models the cost of distinguishing between alternatives, leading to a generalized logit choice model.

While early applications of rational inattention focused primarily on macroeconomic phenomena, its principles have increasingly found application in microeconomic and business contexts. The framework provides theoretical underpinning for consumer behaviors such as non-responsiveness to small price changes and choice inertia \citep{MackowiakWiederholt2009}. Several studies have applied rational inattention to explain strategic pricing behavior, including rigid pricing when consumers are inattentive to price \citep{matvejka2015rigid}, discrete pricing when sellers are inattentive \citep{matvejka2016rationally}, and quality-based strategies when consumers are inattentive to quality \citep{martin2017strategic}. Our contribution to this stream of research is examining the dynamic case when consumers are inattentive to \emph{timing}, a novel yet extremely common and applicable setting in subscription markets. This temporal dimension of inattention has received limited attention in prior work, despite its clear relevance to modern digital commerce.

Recent methodological advances have made rational inattention models more empirically tractable. \citet{joo2023rational} develop a new framework for empirical discrete choice demand estimation based on \citet{MatejkaMcKay2015}'s work, while \cite{brown2024endogenous} introduce a closed-form tractable model of RI-based discrete choice through the elegant assumption of Cardell distribution shocks. Readers interested in empirical applications in marketing and related fields should refer to \citet{turlo2025discrete} for a detailed overview. These developments provide the empirical foundation for testing our theoretical predictions. Empirical researchers can easily extend these existing approaches to our framework by incorporating a dynamic state-space component. 

A major theoretical contribution of our paper is endogenizing the cost of attention. In most extant RI models, the agent's information processing capacity (our $\tau_0$) is treated as an exogenous parameter. In our model, the effective attention sensitivity $\tau(T)$ becomes a direct function of a firm's strategic choice: the trial length $T$. This innovation connects our work to growing interest in how market environments and firm strategies can shape consumer cognition \citep{deClippel2014}. By making the consumer's attention capacity an outcome of the firm's contract design, we create a feedback loop between market structure and cognitive constraints, enabling analysis of how firms might strategically manipulate the informational environment to their advantage.

\subsection{Behavioral Industrial Organization and Exploitation Contracts}

Our paper contributes to the rapidly growing field of behavioral industrial organization, which examines market outcomes when firms interact with consumers who exhibit psychological biases. A significant strand of this literature focuses on the design of "exploitation contracts" that generate profit from consumer mistakes. The seminal work of \citet{GabaixLaibson2006} on markets with shrouded attributes demonstrates how firms may hide the price of add-ons (e.g., printer ink) to exploit "myopic" consumers. Our model shares the theme of firms profiting from a secondary, often overlooked aspect of a product, which in our case is the cognitive cost of cancellation.

However, our conception of the consumer differs fundamentally from models that assume consumer naivety. Rather than positing a fraction of naive consumers, all consumers in our model are sophisticated and forward-looking. They are "rationally inattentive," not naively forgetful. Consumers are fully aware of their own cognitive limitations and anticipate that a long trial will make cancellation more difficult. This anticipation generates what we term "IR-Slack": a counterbalancing loss component that moderates the firm's incentive to fully exploit inattention.

This approach aligns more closely with recent behavioral models that assume consumer sophistication, particularly those dealing with time-inconsistent agents such as those with present-biased preferences \citep{DellaVignaMalmendier2004, HeidhuesKoszegibotond2017}. \citet{DellaVignaMalmendier2004} analyze contracts for gyms and credit cards, demonstrating that firms will optimally design contracts to exploit consumers' self-control problems, even when consumers are sophisticated. Our paper provides a parallel analysis where the friction is inattention rather than present bias, contributing to this literature by proposing a new, attention-based mechanism for consumer inertia and showing how it can be embedded within a contract design problem with sophisticated agents. Recent empirical evidence by \citet{rodemeier2025buy} lends credence to our framework, as that study shows, through a field experiment, that consumers are sophisticated about their inattention in rebate redemption.

This approach also resonates with the behavioral economic literature on "sludge," which examines how firms can increase cognitive friction to guide consumer choices, often to the consumer's detriment \citep{thaler2018nudge}. Our model formalizes these intuitions within a rigorous economic framework, providing testable predictions about when and how firms will deploy such strategies.

Empirical validation of our theoretical framework comes from a recent large-scale field experiment by \citet{Miller2023}. In a study with over two million newspaper readers, they experimentally vary the presence of an automatic renewal clause. Their findings validate the two central, opposing forces in our model. First, they document powerful inertia, finding that a large fraction of consumers who are defaulted into a paid subscription remain subscribed while exhibiting little to no usage, corresponding to the \textit{inattentive revenue} channel we model. Second, and critically, they find strong evidence of consumer sophistication, showing that initial take-up of auto-renewal offers is 24\%--36\% lower than for equivalent auto-cancel offers. This provides direct empirical measurement of the choice-based deterrence that we label the \textit{IR-Slack}. Their estimate that a majority of inert consumers are sophisticated and anticipate their future inertia lends crucial support to our decision to model consumers as rational agents aware of their own cognitive frictions.

\subsection{Digital Economy, Subscription Models, and Customer Management}

The literature on digital economy and subscription models provides a crucial context for our theoretical framework. The traditional explanation for free trials centers on facilitating consumer learning and reducing perceived risk for experience goods \citep{israel2005services, iyengar2007model}. In this sense, trials represent a form of promotional spending to drive adoption. \citet{lee2003new} demonstrate how seeding strategies and free trials can launch new products in the presence of network effects, while \citet{wang2018signaling} show how free trials can signal quality in markets with asymmetric information.

Our model provides a complementary rationale based on Rational Inattention that can explain trial lengths that extend beyond any plausible learning period. The learning and attention mechanisms are not mutually exclusive, and in fact, a complete model of subscription trials would likely include both. However, our contribution is to isolate the attention channel and derive its unique strategic implications, particularly relevant in markets where product quality is easily observable but cancellation requires active effort. As far as we know, ours is the first paper to establish a theoretical connection between subscription model and rational inattention, and the empirical literature on this connection is still highly sparse, with only one recent field experiment by \citet{einav2025selling}.

The freemium model literature offers additional insight into our setting. For example, \citet{pauwels2008moving} and \citet{cao2023free} analyze drivers of conversion from free to paid tiers, while several other studies such as \citet{lee2017designing} and \citet{li2019optimal} examine how firms can optimize the freemium experience to maximize long-term value, by balancing the trade-off between growth and monetization. Our focus on trial length as a strategic variable complements these studies by highlighting an underexplored dimension of freemium strategy. Similarly, research on digital platform engagement \citep{yoganarasimhan2023design} and user stickiness \citep{ray2012research, zhou2024users} provides behavioral foundations for understanding why consumers might remain subscribed despite low usage.

Our work also connects to the extensive literature on customer relationship management (CRM) and customer inertia. It is well-documented that acquiring a new customer is often more expensive than retaining an existing one, leading firms to focus on maximizing customer lifetime value (CLV) \citep{Gupta2004, Venkatesan2007, fader2018customer}. Existing literature has extensively studied switching costs and customer lock-in strategies \citep{Burnham2003, Jones2007, dey2013consumer}. While much of this literature focuses on contractual or technological lock-in \citep{Farrell1989}, our model provides a specific, micro-founded mechanism for psychological lock-in: the inertia arising from the cognitive effort required to terminate a relationship. The strategic implications extend to subscription pricing models more broadly. \citet{danaher2002optimal} study how different pricing structures for usage and access influence users' usage and retention, while \citet{tian2020optimizing} show how firms can use menu design to encourage particular subscription choices. Our model contributes to this literature by focusing on temporal aspects of subscription design, specifically how trial length can be strategically chosen to leverage consumer inattention.

Finally, our analysis of recently proposed "click-to-cancel" regulations and other policy interventions speaks to the growing literature on "dark patterns" in user interface design \citep{Luguri2021, Mathur2019}. These are design choices that nudge users toward outcomes that benefit the firm, but not necessarily the user. A long trial period paired with a complex, multi-step cancellation process can be viewed as a contractual dark pattern. Our model provides a formal economic framework for analyzing the incentives behind such practices and for evaluating the market-wide consequences of policies designed to counter them, enabling policymakers to make informed trade-off decisions.

By formalizing the trade-offs involved, we move beyond purely descriptive accounts of these practices to predictive analysis of how firms will respond to regulation. This provides a valuable tool for both managers seeking to understand competitive dynamics in subscription markets and policymakers evaluating the welfare implications of various regulatory interventions in the digital economy.
%%%%%%%%%%%%%%%%%%%%%%%%%%%%%%%%%%%%%%%%%%%%%%%%%%%%%%%%%%%%%%%%%%%%%%
\section{Model}\label{sec:model}
%%%%%%%%%%%%%%%%%%%%%%%%%%%%%%%%%%%%%%%%%%%%%%%%%%%%%%%%%%%%%%%%%%%%%%

We develop a model of a monopoly subscription platform to analyze how a firm designs its introductory contract, comprising a trial length and a post-trial price, for a market of consumers who are fully rational but subject to attentional frictions. Our framework isolates the role of trial length as a strategic tool for managing consumer attention, distinct from its traditional role in facilitating learning. The model is structured to be rich enough to capture the essential behavioral trade-offs while remaining tractable enough to yield sharp, testable predictions.

\subsection{Primitives and Timeline}\label{sec:primitives}

\paragraph{Consumers and Valuations.} We consider a market with a unit mass of consumers, each characterized by a per-period valuation $v$ for the platform's service. This valuation is the consumer's private information and is drawn from a cumulative distribution function $F(v)$ with continuous density $f(v)>0$ on the support $[0,1]$. An important modeling choice we make, a departure from the prior literature on trial offers, is that consumers know their product valuation $v$ from the outset. This assumption allows us to cleanly isolate the role of attention from the role of learning. In a learning model \citep{iyengar2007model}, the trial period's primary function is to allow consumers to resolve uncertainty about $v$. By assuming $v$ is known, we create a setting where, from a purely informational perspective, no trial period is necessary. The existence of a non-zero optimal trial length in our model will therefore be attributable solely to attentional frictions, providing a complementary and so far overlooked rationale for this ubiquitous product marketing tool.

\begin{table}[!htbp]
\centering
\caption{Summary of Model Notation}
\label{tab:notation}
% Use tabularx environment to set a fixed total width
\begin{tabularx}{\textwidth}{@{} l c X @{}} 
\toprule
\textbf{Symbol} & \textbf{Type} & \textbf{Definition} \\
\midrule
\multicolumn{3}{l}{\textit{\textbf{Consumer Primitives}}} \\
$v$ & Parameter & Consumer's per-period valuation for the service, $v \sim F$ on $[0,1]$ \\
$F(v), f(v)$ & Parameter & CDF and PDF of consumer valuations \\
$\tau_0$ & Parameter & Baseline attention sensitivity (at $T=0$) \\
$\beta$ & Parameter & Rate of attention decay over time \\
\midrule
\multicolumn{3}{l}{\textit{\textbf{Firm / Contract Primitives}}} \\
$T$ & Choice Variable (Firm) & Length of the free trial period, $T \in \mathbb{N}$ \\
$P$ & Choice Variable (Firm) & Post-trial automatic renewal price, $P \in (0,1]$ \\
\midrule
\multicolumn{3}{l}{\textit{\textbf{Consumer Choice and Attention}}} \\
$q$ & Choice Variable (Consumer) & Probability of successfully remembering to cancel \\
$\tau(T)$ & Outcome Variable & Effective attention sensitivity, $\tau(T) = \tau_0 / (1+\beta T)$ \\
$q^*$ & Outcome Variable & Optimal cancellation probability chosen by the consumer \\
\midrule
\multicolumn{3}{l}{\textit{\textbf{Key Aggregate Outcomes}}} \\
$\Pi(T,P)$ & Outcome Variable & Firm's total profit \\
$\IR(T,P)$ & Outcome Variable & Inattentive Revenue: profit from users who forget to cancel \\
$U(T,P)$ & Outcome Variable & Ex-ante expected consumer utility from the contract \\
$\IRSlack(T,P)$ & Outcome Variable & Marginal harm to consumer utility from increasing $T$ \\
\midrule
\multicolumn{3}{l}{\textit{\textbf{Refinement Assumption Parameters}}} \\
$\varepsilon$ & Parameter (A1) & Constant price elasticity of "happy subscriber" demand \\
$\kappa$ & Parameter (A1) & Scaling constant for iso-elastic demand \\
$h(v)$ & Property (A2) & Hazard rate of the valuation distribution, $f(v)/(1-F(v))$ \\
\midrule
\multicolumn{3}{l}{\textit{\textbf{Extension: Paid Trials (Section 6)}}} \\
$P_0$ & Choice Variable (Firm) & Upfront introductory price for the trial period \\
$\eta(P_0)$ & Function & Fraction of consumers who sign up at intro price $P_0$ \\
$\varepsilon_0(P_0)$ & Outcome Variable & Elasticity of the sign-up rate with respect to $P_0$ \\
$P^{\text{aug}}$ & Outcome Variable & Total expected post-trial profit per subscriber \\
$\theta$ & Parameter (A3) & Constant elasticity of sign-up demand \\
\bottomrule
\end{tabularx}
\end{table}

\paragraph{The Firm and the Contract.} A monopolist firm offers a subscription contract defined by a pair $(T,P)$, where $T \in \mathbb{N}$ is the length of the free trial period (e.g., in days or weeks), and $P \in (0,1]$ is the price charged automatically upon renewal at the end of the trial. The firm's marginal cost of production is normalized to zero, a common approximation for digital goods. We ignore time discounting within the initial subscription cycle for simplicity. The monopoly setting allows us to focus purely on the firm's contract design incentives without the confounding effects of competition, representing markets for highly differentiated products (e.g., specialized SaaS like Salesforce) or platforms with dominant market positions (e.g., Adobe Creative Cloud).

\paragraph{Sequence of Events.} The strategic interaction unfolds over a clear sequence of stages, which mirrors the typical customer journey for a subscription service:
\begin{itemize}[leftmargin=1.8em]
    \item \textbf{Stage 0 (Contract Announcement):} The firm chooses and publicly commits to a contract $(T,P)$. This commitment is essential; the firm cannot change the terms ex-post for consumers who have signed up.
    \item \textbf{Stage 1 (Sign-up Decision):} Each consumer observes the contract terms $(T,P)$. Knowing their own valuation $v$, they form expectations about the future costs and benefits of the contract, including the cognitive burden of managing the subscription. Based on this forward-looking assessment, they decide whether to sign up.
    \item \textbf{Stage 2 (Monitoring Decision):} During the trial period, a consumer with a low valuation ($v<P$) who intends to cancel must decide on a monitoring strategy. This is an abstract choice representing the level of costly cognitive effort they will exert to set up a "reminder" to cancel the subscription before it renews. This could involve literally setting a calendar alert, placing a sticky note on a monitor, or simply trying to keep the task "top of mind."
    \item \textbf{Stage 3 (Renewal):} At the end of the trial, the subscription automatically renews at price $P$ unless the consumer has successfully executed a cancellation action.
\end{itemize}

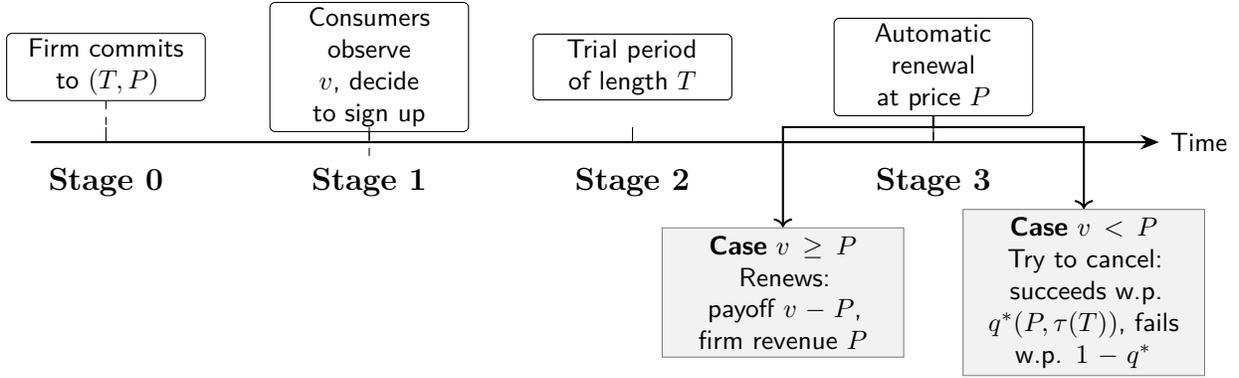
\begin{figure}[h!]
\centering
\begin{tikzpicture}[
    font=\sffamily\footnotesize,
    timeline/.style={-{Stealth[length=2.5mm,width=2mm]}, thick},
    stage/.style={anchor=north,font=\bfseries},
    event/.style={rectangle, draw=black, rounded corners=2pt,
                  text width=2.5 cm, align=center, inner sep=2 pt},
    outcome/.style={rectangle, draw=black!50, fill=gray!10,
                    text width=3 cm, align=center, inner sep=3pt},
  ]
  % timeline axis
  \draw[timeline] (0,0) -- (15,0) node[anchor=west] {Time};
  % stage ticks & labels
  \foreach \x/\label in {1/Stage~0, 4.5/Stage~1, 8/Stage~2, 12/Stage~3}
    {
      \draw (\x,0) -- (\x,0.2);
      \node[stage] at (\x,-0.2) {\label};
    }
  % Events
  \node[event] (design) at (1,1) {Firm commits to $(T,P)$};
  \node[event] (signup) at (4.5,1) {Consumers observe $v$, decide to sign up};
  \node[event] (trial)  at (8,1) {Trial period of length $T$};
  \node[event] (renew)  at (12,1) {Automatic renewal at price $P$};
  % Dashed connectors
  \foreach \n in {design,signup,renew}
    \draw[dashed] (\n.south) -- ++(0,-0.4);
  % Outcomes from renewal
  \node[outcome] (keep) at (10,-2)
    {\textbf{Case $v\ge P$}\\ Renews:\\ payoff $v-P$, firm revenue $P$};
  \node[outcome] (cancel) at (14,-2)
    {\textbf{Case $v<P$}\\ Try to cancel: \\ 
     succeeds w.p. $q^{*}(P,\tau(T))$, fails w.p. $1-q^{*}$};
  % Arrows from renewal
  \draw[->,thick] (renew) -- ++(0,-0.8) -| (keep.north);
  \draw[->,thick] (renew) -- ++(0,-0.8) -| (cancel.north);
\end{tikzpicture}
\caption{Timeline of the subscription contract game}
\label{fig:timeline}
\end{figure}

\subsection{Minimal Assumptions (MA)}\label{sec:MA}
Our core mechanism is built on a minimal set of five assumptions that define the consumer's environment and decision-making process. These are designed to be general yet capture the essential features of the problem.

\begin{enumerate}[label=\textbf{MA-\arabic*}, leftmargin=2.3em]
    \item\label{MA1} \textbf{Valuation Support:} The valuation density $f(v)$ is positive and continuous on $[0,1]$. This is a standard technical assumption that ensures a heterogeneous market with active demand at all relevant prices, making the firm's pricing and contract design problem non-trivial.

    \item\label{MA2} \textbf{Information Cost (Shannon):} The cognitive cost to achieve a cancellation probability of $q \in (0,1)$ is given by $C(q;\tau)=\frac{1}{\tau}\mathcal{I}(q)$, where $\mathcal{I}(q)=q\ln q+(1-q)\ln(1-q)$ is the Shannon Mutual Information function and $\tau>0$ is the consumer's attention sensitivity. This widely used formulation, grounded in information theory and pioneered in economics by \citet{Sims2003}, provides a rigorous micro-foundation for costly information processing. It has been successfully applied in numerous contexts, from macroeconomics to marketing, and its key feature is the convexity of the cost function \citep{MatejkaMcKay2015}. This convexity captures the intuitive idea of increasing marginal costs to achieving perfect certainty: it is relatively easy to improve one's chances of remembering from 50\% to 60\%, but it requires immense effort to improve from 99\% to 100\%.

    \item\label{MA3} \textbf{Memory Decay:} The consumer's attention sensitivity $\tau$ is a decreasing function of the trial length $T$, given by the specific functional form $\displaystyle\tau(T)=\frac{\tau_{0}}{1+\beta T}$, with $\tau_{0}>0$ and $\beta>0$. Here, $\tau_0$ represents the consumer's baseline attention for an immediate task, which can be seen as an individual trait (e.g., being naturally organized) or as being influenced by the technological environment (e.g., access to sophisticated calendar apps). The parameter $\beta$ captures the rate of attention decay over time. This assumption is grounded in the psychology framework of \emph{prospective memory}, which documents that the ability to remember to perform a delayed intention decays over time \citep{mcdaniel2000strategic, Kliegel2002}. The hyperbolic form is chosen for its tractability and its ability to elegantly nest the standard full-attention model (when $\beta=0$)\footnote{Hyperbolic decay is also consistent with the common model of delay discounting in behavioral economics \citep{ainslie1991derivation}, or the neuroscience finding about animal integration of past experiences \cite{danskin2023exponential}.}. When $\beta>0$, the trial length $T$ becomes a strategic instrument for the firm to directly and predictably influence the cognitive friction faced by its customers. In Appendix F, we provide a robustness analysis of alternative functional forms of memory decay.

    \item\label{MA4} \textbf{Cancellation Rule:} If a consumer with $v<P$ successfully remembers to cancel (which occurs with their chosen probability $q$), they do so and avoid the charge. Otherwise, due to inattention, they remain subscribed for the period and pay the price $P$. This represents the "automatic" nature of the renewal.

    \item\label{MA5} \textbf{Outside Option:} The utility from not subscribing, or from successfully canceling, is normalized to zero. This is a standard normalization that simplifies the analysis without loss of generality.
\end{enumerate}
\subsection{Refinement Assumptions for Tractability}\label{sec:extra-assump}
To derive sharp, analytical predictions that are both parsimonious and empirically relevant, we introduce two standard refinement assumptions on the structure of demand. These assumptions are common in the industrial organization and marketing literature and will be invoked specifically in Section 5 to obtain our main comparative static results.

\begin{assumption}[A1: Isoelastic Happy-Subscriber Demand]\label{ass:A1}
For renewal prices $P$ in the empirically relevant range, the survivor function of valuations is isoelastic: $1-F(P)=\kappa P^{-\varepsilon}$ with constants $\kappa>0$ and $\varepsilon\in(0,1)$.
\end{assumption}
\paragraph{Economic Interpretation.} This assumption posits that once a consumer’s valuation clears the hurdle of being a potential subscriber, their willingness-to-pay exhibits a power-law tail. This is a common feature of demand systems for established goods and is consistent with the observation that many mature subscription services price on the inelastic portion of their demand curve. A constant elasticity of demand arises, for example, from a population of consumers with Cobb-Douglas utility functions. It implies that each 1\% price increase shrinks the segment of "happy subscribers" (those with $v \ge P$) by a constant $\varepsilon$ percent. As we will show, this structure greatly simplifies the firm's marginal revenue calculation, allowing for a clean analytical characterization of the optimal price and its response to policy shocks.

\paragraph{Real-World Intuition.} Streaming services provide a compelling illustration of isoelastic tails. For example, after Netflix’s January 2022 price hike from \$13.99 to \$15.49 (an 11\% increase), its subscriber churn rate increased from 2.3\% to 3.3\%, implying a demand elasticity of approximately 0.4.\footnote{Data from Antenna Research, as reported in industry analyses of \href{https://www.antenna.live/insights/netflix-q122-retrospective}{Netflix's Q1-2022 performance}.} Similarly, Spotify’s Premium tiers have exhibited monthly churn in the 1.5\% to 2.0\% range across multiple price points and markets, consistent with an inelastic demand response.\footnote{See, for example, Inderes Equity Research, “Spotify: churn and pricing” (2023).} With $\varepsilon<1$ in these prominent cases, our model's prediction that policies like "click-to-cancel" will raise renewal prices modestly while sharply shortening trial lengths becomes directly relevant.

\begin{assumption}[A2: Low-Value Mass and Increasing Failure Rate]\label{ass:C1}
The valuation distribution $F$ has a positive mass of low-value users ($F(\underline v)>0$ for some $\underline v>0$) and a weakly increasing hazard rate, $h(v)=f(v)/(1-F(v))$, on the interval $[0,1]$.
\end{assumption}
\paragraph{Economic Interpretation.} This assumption serves two crucial roles. First, the existence of a low-value segment ensures that there is always a pool of consumers who need to cancel, making the cancellation problem economically relevant. It also ensures that the consumer participation constraint will eventually bind as the trial length grows and the expected costs of managing the subscription become prohibitive for these low-value types. Second, the Increasing Failure Rate (IFR) property is a standard regularity condition in contract theory and industrial organization that guarantees the firm's profit function is well-behaved (specifically, concave in price), guaranteeing a unique optimal solution \citep{lariviere2001selling}. It ensures that the marginal revenue from "happy" subscribers diminishes sufficiently quickly, which is necessary for a well-defined optimum.  This assumption is weaker than assuming a specific parametric family; it requires only that the ``conditional dropout'' rate $h(v)$ does not decline.  

\paragraph{Real-World Intuition.} The existence of a low-value segment is empirically well-documented. In every real subscription market, there are some “dabblers” who try the service and quickly realize it is not worth paying for (e.g., people who sign up for a fitness app once and never open it again). Survey panels for digital news, fitness apps, and other services consistently record 10–20\% of trial users reporting a willingness-to-pay near zero after the trial. These are the users for whom the cancellation decision is most critical. Furthermore, most common distributions used to fit empirical willingness-to-pay data, such as the Weibull, log-normal, and truncated normal distributions, exhibit the IFR property. Together, these conditions ensure our model captures a key real-world friction: firms face a genuine risk of trial rejection and loss of goodwill if the cancellation process is perceived as too burdensome, which is precisely the mechanism that disciplines the firm's use of long trials in our model.
% =====================================================================
%  Information–Cost Paradigm  (Matejka–McKay convention, exact notation)
% =====================================================================

\subsection{Information–Cost Paradigm \citep{MatejkaMcKay2015}}\label{sec:info}
To model the consumer's cancellation decision, we employ the rational inattention (RI) framework based on Shannon entropy, as formalized for discrete choice by \citet{MatejkaMcKay2015}. This approach posits that processing information to move from a state of uncertainty to a state of certainty is cognitively costly.

\paragraph{Shannon entropy and mutual information.}
For a binary random variable that takes value $C\in\{\text{Cancel},\text{No‐Cancel}\}$ with probabilities $(q,1-q)$,
the Shannon entropy (in nats) is
\[
H_{\mathrm{bin}}(q)\;=\;-\Bigl[q\ln q + (1-q)\ln(1-q)\Bigr]\;\;\ge 0,
\qquad q\in(0,1).
\]
Assuming the customer starts from an uninformative prior over the cancellation $p=(\tfrac12,\tfrac12)$\footnote{In Appendix F, we show that our results hold in the case of informative priors. The uninformative prior assumption is for brevity.}, following the Rational Inattention framework \citep{Sims2003}, we define the "amount" of information that has been used by an customer in making the optimal decision as the \textit{mutual information} processed by them, or equivalently the Kullback–Leibler divergence of $(q,1-q)$ from $p$, is
\begin{align}
\mathcal{I}(q\parallel p)
  &= q\ln\!\Bigl(\tfrac{q}{1/2}\Bigr)
     +(1-q)\ln\!\Bigl(\tfrac{1-q}{1/2}\Bigr)                      \nonumber\\
  &= q\ln q+(1-q)\ln(1-q) +\ln 2                                  \label{eq:MI}\\
  &= -H_{\mathrm{bin}}(q) \;+\; \ln 2.                            \nonumber
\end{align}
It satisfies $\mathcal{I}(q\parallel p)\ge 0$ with equality at
$q=\tfrac12$ (no attention).

\begin{definition}[Shannon mutual–information cost]\label{def:cost_MI}
Given an \emph{attention sensitivity} $\tau>0$, the cognitive cost of
implementing reminder probability $q$ is
\[
C(q;\tau)
  \;=\;\frac{1}{\tau}\,
        \mathcal{I}(q\parallel\tfrac12).
\tag{C‐MI}\label{eq:cost_MI}
\]
Higher $\tau$ makes information processing cheaper.\footnote{Here, we use $\tau$ to denote sensitivity, since that is a main focus of our model. Readers familiar with RI framework will realize this is the reciprocal of $\lambda$, the unit cost of attention, in \citet{MatejkaMcKay2015}.}
\end{definition}

Because $\ln 2/\tau$ is a \emph{constant} w.r.t.\ $q$, we may drop it
without affecting any optimization.  Accordingly, throughout the proofs
we use the algebraically lighter equivalent
\[
C(q;\tau)=\frac{1}{\tau}\Bigl[q\ln q+(1-q)\ln(1-q)\Bigr].
\tag{C}\label{eq:cost_simplified}
\]

The Shannon Mutual information function has several properties that make it a compelling measure of cognitive friction. The parameter $\tau$ is central, acting as an inverse measure of the "price of attention." A consumer with a high $\tau$ may be naturally well-organized, have access to better reminder technology (e.g., sophisticated calendar apps), or face a simpler, less onerous cancellation task; for them, achieving certainty is relatively cheap. The cost function $C(q;\tau)$ itself is convex in $q$, capturing the crucial economic intuition of increasing marginal effort. For instance, moving from complete uncertainty ($q=0.5$) to moderate certainty ($q=0.6$) requires less cognitive effort than moving from high certainty ($q=0.99$) to near-perfect certainty ($q=0.999$). This reflects the real-world difficulty and escalating mental burden of eliminating all possibility of error. Furthermore, the cost is minimized and equal to zero only when no new information is processed and the consumer remains at their uninformative prior ($q=0.5$). This specific functional form is not only tractable, i.e., leading directly to the widely-used logit choice model, but is also deeply rooted in the information-theoretic foundations of the RI literature \citep{Sims2003}, providing a disciplined micro-foundation for modeling bounded rationality.

% ---------------------------------------------------------------------
\subsection{Consumer Monitoring Problem}\label{sec:consumer}

To model the cancellation decision of a consumer with $v < P$, we employ the rational inattention framework. The consumer faces a binary choice: `Cancel` or `Do Not Cancel`. The `Cancel` action is optimal and yields a payoff of 0. The `Do Not Cancel` action is a mistake that results in a payment of $P$. The consumer can reduce the probability of making this mistake by devoting cognitive resources (e.g., setting reminders, mentally rehearsing the task) to monitoring the deadline. Let $q\in[0,1]$ be the probability that the consumer successfully remembers to cancel. Achieving a high value of $q$ is cognitively costly. The consumer decision problem is
\[
\min_{q\in[0,1]}
   \;\;\Bigl[(1-q)P\Bigr]
   \;+\;
   \frac{1}{\tau(T)}\,\mathcal{I}\!\bigl(q\parallel\tfrac12\bigr),
\tag{1$'$}\label{eq:consumer_obj_MI}
\]

or, dropping the constant $\ln 2/\tau(T)$, this reduces to,
\[
\min_{q\in[0,1]}
   \;\;(1-q)P+\frac{1}{\tau(T)}\Bigl[q\ln q+(1-q)\ln(1-q)\Bigr].
\tag{1}\label{eq:consumer_obj_clean}
\]

\begin{lemma}[Optimal monitoring (RI logit)]\label{lem:qstar}
The objective in \eqref{eq:consumer_obj_clean} is strictly convex on
$(0,1)$ and attains its unique minimum at
\[
q^{*}(P,\tau)
   \;=\;\frac{1}{1+\exp(-\tau P)}.
\]
Moreover, the cancellation probability $q^*$ is increasing in price $P$ and attention $\tau$, but decreasing in trial length $T$:
\[
\frac{\partial q^{*}}{\partial P}= \tau q^{*}(1-q^{*})>0,\quad
\frac{\partial q^{*}}{\partial \tau}= P\,q^{*}(1-q^{*})>0,\quad
\frac{\partial q^{*}}{\partial T}= -\beta\tau_{0}P
       \frac{q^{*}(1-q^{*})}{(1+\beta T)^{2}}<0.
\]
\end{lemma}

\begin{proof}[Sketch]
Take the derivative of \eqref{eq:consumer_obj_clean}:
\(
\partial_q L= -P+\tau^{-1}\ln\!\bigl(q/(1-q)\bigr).
\)
Setting it to zero yields
$\ln\!\bigl(q/(1-q)\bigr)=\tau P$, i.e. the stated logit.
Strict convexity follows because
$\partial_{qq}L=\tau^{-1}/[q(1-q)]>0$.
The comparative statics are direct differentiations,
and $\partial_T q^{*}$ uses the chain rule with
$d\tau/dT=-\beta\tau_{0}/(1+\beta T)^{2}$.
\end{proof}

The results of this lemma are intuitive. Consumers will try harder to remember to cancel (i.e., choose a higher $q^*$) when the financial stakes are higher (a larger $P$) and when paying attention is cognitively cheaper (a higher $\tau$). The result that $q^*$ decreases with trial length $T$ is a direct consequence of our memory decay assumption (MA-3) and forms the central behavioral mechanism of our model. The logistic functional form is a hallmark of the Shannon entropy cost structure and provides an analytically tractable foundation for the rest of our analysis.

\subsection{Aggregate Outcomes and Firm Objective}\label{sec:agg}
With the individual consumer's behavior defined, we can now specify the aggregate market outcomes and the firm's optimization problem.

\paragraph{Inattentive-Revenue Component.}
The firm's total revenue comes from two distinct sources. The first is standard revenue from "willing subscribers" with $v \ge P$ who value the service above its price. The second, and the focus of our behavioral model, is revenue from "inattentive subscribers." These are consumers with $v < P$ who optimally would have canceled but who failed to do so due to attentional friction. We define the revenue from this second group as \textit{inattentive revenue}.

\begin{equation}
\boxed{\;
\text{IR}(T,P)=
P\!\int_{0}^{P}\!\bigl[\,1-q^{*}(P,\tau(T))\bigr]f(v)\,dv
\;}
\tag{2}\label{eq:IR}
\end{equation}
Since the cancellation probability $q^*$ is constant for all consumers with $v<P$, this expression simplifies to $\text{IR}(T,P) = P \cdot F(P) \cdot [1-q^*(P, \tau(T))]$. This component of profit is central to our analysis, as it is the channel through which the firm's choice of $T$ affects its revenue.

\paragraph{Total Profit.} Hence, we define the firm's total profit as the sum of standard revenue from willing subscribers, as in standard economic models, and inattentive revenue from (rationally) ``inattentive subscribers''.
\begin{equation}
\Pi(T,P)\;=\;
\underbrace{P\bigl[1-F(P)\bigr]}_{\text{Standard Revenue}}\;+\;\underbrace{\IR(T,P)}_{\text{Inattentive Revenue}}.
\tag{3}\label{eq:profit}
\end{equation}

\paragraph{Expected Consumer Utility.} A forward-looking consumer's ex-ante expected utility from accepting the contract is the sum of the expected surplus for high-valuation types and the expected losses for low-valuation types. The losses for the latter group include both the expected monetary loss from failing to cancel and the cognitive cost incurred in trying to remember.
\begin{equation}
\boxed{
\;
U(T,P)
  \;=\;
  \underbrace{\int_{P}^{1}\!\bigl(v-P\bigr)\,f(v)\,dv}_{\text{Surplus of “happy’’ subscribers}}
  \;-\;
  \underbrace{P\,F(P)\,\bigl[1-q^{*}(P,\tau(T))\bigr]}_{\text{Expected monetary loss from forgetting}}
  \;-\;
  \underbrace{\frac{1}{\tau(T)}\,
              I\!\bigl(q^{*}(P,\tau(T))\bigr)\,F(P)}_{\text{Cognitive cost of monitoring}}
\;}
\label{eq:utility}
\end{equation}

The first term, $\int_{P}^{1}(v-P)f(v)dv$, represents the standard consumer surplus enjoyed by "happy subscribers." These are consumers whose valuation $v$ is greater than or equal to the price $P$. They intend to use the service post-trial and derive a net benefit from doing so. The next two terms capture the expected outcome for the "unhappy subscribers" (or more accurately, the "unwilling subscribers"), those with $v < P$ who would prefer to cancel. This term is itself composed of two distinct sources of disutility. The first component, $-P\,F(P)\,\bigl[1-q^{*}(P,\tau(T))\bigr]$, is the expected monetary loss. It is the price $P$ multiplied by the probability of failing to cancel, aggregated over all consumers in this segment. The second component, $-\frac{1}{\tau(T)}\,I\!\bigl(q^{*}(P,\tau(T))\bigr)\,F(P)$, represents the aggregate cognitive cost of attention. This is a crucial feature of our model: even consumers who successfully cancel (and thus pay nothing) still incur a utility loss from the mental effort required to ensure they remembered to do so.

This comprehensive utility function means that consumers in our model are sophisticated and forward-looking. When deciding whether to sign up (Stage 1), they do not just consider the price $P$; they also anticipate the future hassles. The sophisticated consumers understand that a long trial (high $T$) or a high price (high $P$) will increase both their risk of a costly mistake and the cognitive burden of avoiding it. For consumers to be willing to sign up for the trial, this aggregate expected utility must be non-negative. This gives rise to an individual rationality (IR) or participation constraint, $U(T,P)\ge0$, which is the central constraint on the firm's behavior.

\begin{definition}[IR-Slack]\label{def:IR-Slack}
We define the \textbf{IR-Slack} as the marginal harm to consumer utility from extending the trial period. To derive this, we apply the envelope theorem to the consumer's cost-minimization problem in Equation \eqref{eq:consumer_obj_MI}. The derivative of the minimized loss $L^* = (1-q^*)P + C(q^*;\tau)$ with respect to $\tau$ is:
\[
\frac{\partial L^*}{\partial \tau} = \frac{\partial}{\partial \tau}\left(\frac{1}{\tau}\mathcal{I}(q^* \parallel 1/2)\right) = -\frac{1}{\tau^2}\mathcal{I}(q^* \parallel 1/2).
\]
The derivative of aggregate utility with respect to $T$ is then:
\[
\frac{\partial U}{\partial T} = \frac{\partial U}{\partial \tau}\frac{d\tau}{dT} = \left(-F(P)\frac{\partial L^*}{\partial \tau}\right)\frac{d\tau}{dT} = \left(F(P)\frac{1}{\tau^2}\mathcal{I}(q^* \parallel 1/2)\right) \left(-\frac{\beta\tau^2}{\tau_0}\right) = -\frac{\beta}{\tau_0}F(P)\mathcal{I}(q^* \parallel 1/2).
\]
The IR-Slack is the negative of this quantity:
\[
\IRSlack(T,P) := -\frac{\partial U(T,P)}{\partial T} = \frac{\beta}{\tau_{0}}F(P)\,\mathcal{I}(q^* \parallel 1/2) > 0.
\tag{5}\label{eq:slack_revised}
\]
\end{definition}

The above is a critical counterbalancing force that constrains the firm's incentive to extend trials too long. The consumers in our model are sophisticated; they are not passive victims of their own inattention but are forward-looking agents who anticipate the cognitive costs and potential for error associated with managing a long-term commitment. When a firm increases the trial length $T$, consumers recognize that this makes the contract more difficult and costly to manage. This anticipated "hassle cost" reduces their ex-ante valuation of the firm's offer.

The IR-Slack measures this marginal erosion of consumer utility. A firm that ignores this effect will find its offer becomes unattractive, leading to lower sign-up rates and a smaller customer base. This finding aligns perfectly with the large-scale field experiment by \citet{Miller2023}, who find that auto-renewal offers (which have a higher perceived hassle cost) see significantly lower take-up rates than auto-cancel offers. The IR-Slack, therefore, represents the cost side of the firm's trade-off, a cost paid in the currency of consumer goodwill and participation. The formula is intuitive: the marginal harm to consumers is proportional to the rate of attention decay ($\beta/\tau_0$), the cognitive effort required to overcome it (captured by $I(q^*)$), and the size of the population that is vulnerable to this friction ($F(P)$).

\paragraph{Final Firm Objective.} The platform's objective is to maximize its profit, $\Pi(T,P)$, as defined in Equation \eqref{eq:profit}, subject to the consumer participation constraint, $U(T,P)\ge0$. Because consumer utility $U$ is decreasing in both $P$ and $T$ (through the IR-Slack), this constraint will bind at the optimum for a monopolist seeking to extract maximum value. A firm that makes its offer too onerous (either too high a price or too long a trial) will find its customer base vanishes. This is a standard feature in models with sophisticated consumers who must be induced to participate. Therefore, the firm's problem is:
\[
\max_{T,P}\;\; \Pi(T,P)
\quad
\text{such that}\quad U(T,P)\ge 0.
\tag{6}\label{eq:firmproblem}
\]
Because \(U(T,P)\) is decreasing in both \(T\) and \(P\),
the constraint binds at the optimum.  Letting \(\mu>0\) denote the
Lagrange multiplier, the Lagrangian is
\(\mathcal{L}=\Pi(T,P)-\mu\,U(T,P)\).

This constrained optimization problem defines the optimal contract $(T^*, P^*)$ and forms the basis for all subsequent analysis. The binding constraint implies that the firm will design the contract to extract all consumer surplus, setting the terms such that the ex-ante utility for the marginal consumer (or the average consumer, depending on the setup) is exactly zero.

%%%%%%%%%%%%%%%%%%%%%%%%%%%%%%%%%%%%%%%%%%%%%%%%%%%%%%%%%%%%%%%%%%%%%%
\section{Functional Form Free Results}\label{sec:general}
%%%%%%%%%%%%%%%%%%%%%%%%%%%%%%%%%%%%%%%%%%%%%%%%%%%%%%%%%%%%%%%%%%%%%%

In this section, we derive the foundational results of our model. These propositions are "functional-form-free" in the sense that they do not depend on the specific Shannon entropy cost function, but rather on a minimal set of qualitative properties that any reasonable model of costly attention should exhibit. Specifically, we require only that (a) the consumer's ability to pay attention, $\tau(T)$, decreases as the task's deadline, $T$, moves further into the future, and (b) the consumer's chosen probability of success, $q^*$, increases with their ability to pay attention, $\tau$. These general results reveal the main economic tensions that drive the firm's strategic decisions and provide a robust theoretical backbone for the sharper, more specialized predictions that will follow.

\begin{proposition}[Attention vs. Trial Length]\label{prop:F1}
A longer trial period unambiguously decreases the consumer's effort to cancel. For any given renewal price $P>0$, if the rate of attention decay is positive ($\beta > 0$), then an increase in trial length $T$ leads to a strict decrease in the optimal cancellation probability $q^*(P, \tau(T))$.
\[
\frac{\partial q^{*}(P,\tau(T))}{\partial T}\;<\;0 .
\]
\end{proposition}
\begin{proof}
The proof is a direct application of the chain rule. The change in the optimal cancellation probability with respect to the trial length is given by $\frac{\partial q^*}{\partial T} = \frac{\partial q^*}{\partial \tau} \frac{d\tau}{dT}$. From the consumer optimization problem in \Cref{lem:qstar}, a higher attention sensitivity $\tau$ makes it cheaper for them to exert mental effort, leading to a higher chosen success probability, so $\frac{\partial q^*}{\partial \tau} > 0$. From our core behavioral assumption on attention decay, \Cref{MA3}, a longer trial reduces attention sensitivity, so $\frac{d\tau}{dT} < 0$ for any $\beta>0$. The product of a positive term and a negative term is necessarily negative, thus complete the proof.
\end{proof}

\begin{figure}
    \centering
    \includegraphics[width=0.5\linewidth]{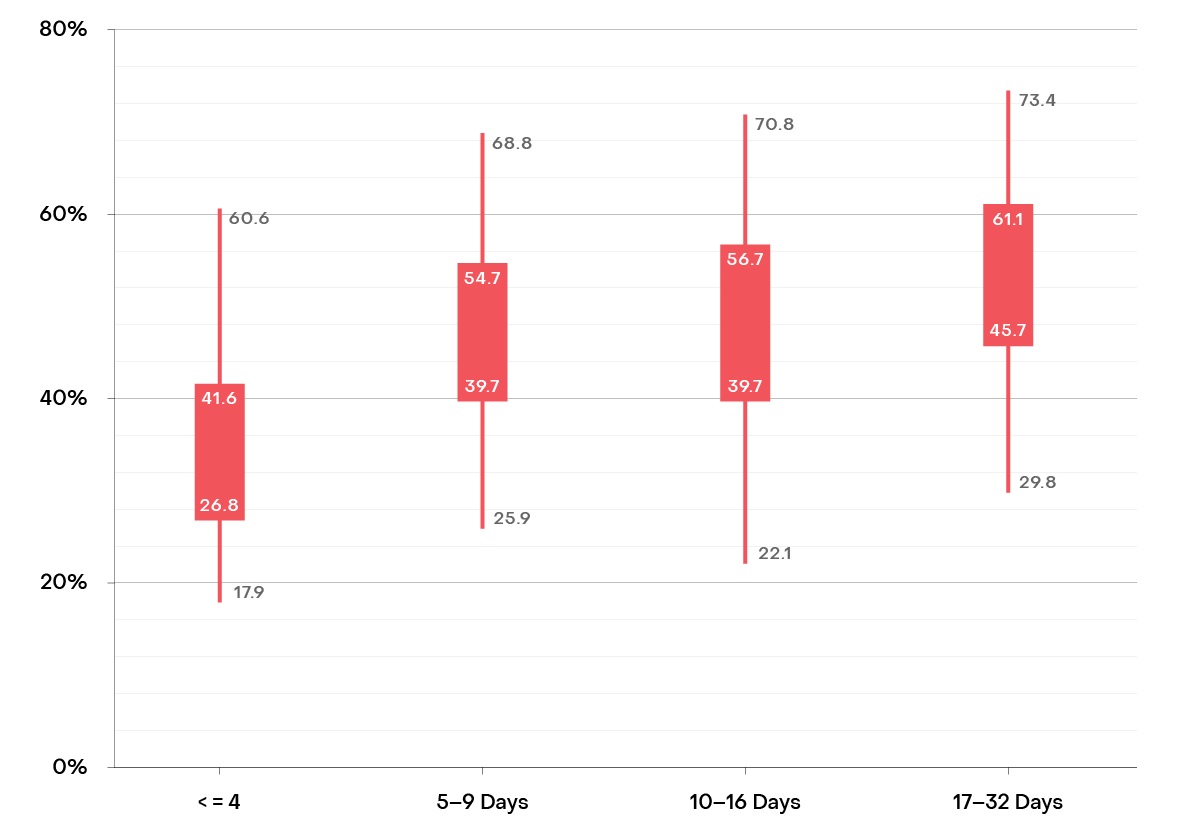}
    \caption{Conversion rate by trial length for mobile apps (Source: RevenueCat).}
    \label{fig:conversiontrial}
\end{figure}

This proposition describes the primary behavioral channel through which the firm's strategic plan design influences consumer outcomes. The trial length $T$ is not just a passive waiting period, it is an active strategic instrument that shapes the consumer's cognitive environment. By extending the trial, the firm makes the future cancellation task less salient and more susceptible to being forgotten. This effectively raises the cognitive price of remembering. The rationally inattentive consumer, who is aware of this increased difficulty, optimally responds by reducing their monitoring effort, which manifests as a lower probability of successful cancellation. This result provides the firm with a direct and predictable approach to influence the rate of what we term "accidental renewals.". This analytical result also fits well with industry reports, such as RevenueCat, a top revenue management platform for subscription apps\footnote{\href{https://www.revenuecat.com/state-of-subscription-apps-2025/}{State of Subscription Apps 2025.}}, as shown in Figure \ref{fig:conversiontrial}.

\begin{proposition}[Inattentive Revenue Monotonicity]\label{prop:F2}
The firm's inattentive revenue, $\text{IR}(T,P)$, as defined in Equation \ref{eq:IR}, is weakly increasing in the trial length $T$. If $\beta > 0$ and there is a positive mass of consumers with valuations below the price ($F(P)>0$), it is strictly increasing.
\end{proposition}
\begin{proof}
By definition, inattentive revenue is $\text{IR}(T,P) = P \cdot F(P) \cdot [1-q^*(P, \tau(T))]$. Holding the price $P$ constant, the only term in this expression that depends on the trial length $T$ is the cancellation probability $q^*$. From \Cref{prop:F1}, we know that $\partial q^*/\partial T < 0$. This implies that the probability of \textit{failing} to cancel, which is $1-q^*$, must be strictly increasing in $T$. Since $P>0$ and we assume $F(P)>0$, the entire expression for $\text{IR}$ must be strictly increasing in $T$.
\end{proof}

Intuitively, this rationalizes the firm's primary incentive to offer long trial periods, an incentive cannot be fully explained by consumer learning. From a purely revenue-focused perspective, a longer trial is always better. It systematically increases the number of consumers who, despite preferring to cancel, fail to do so and are subsequently charged the renewal price. This finding establishes an attention-based explanation for a key stylized fact: firms frequently offer trial periods (e.g., 30-60 days) that far exceed the time necessary to evaluate the product's features. If inattentive revenue were the only consideration and consumers signed up regardless of the burden of the trial, the firm’s profit would rise monotonically as the trial got longer, and the firm would be incentivized to extend trial length until the chance of forgetting is 100\%. In reality, a longer trial also makes the offer less attractive ex ante, captured by the IR-Slack term, so the firm eventually loses participants. Balancing these two forces yields a finite, moderate optimal trial length.

\begin{proposition}[Existence of an Interior Optimal Trial Length]\label{prop:F3}
If the attention decay rate is positive ($\beta>0$) and there is a positive mass of consumers who wish to cancel (per \Cref{ass:C1}), there exists a finite optimal trial length \(T^{*}\in(0,\infty)\)that solves the firm's optimization problem. This $T^*$ satisfies the first-order condition:
\[
\frac{\partial \mathcal{L}}{\partial T} = \frac{\partial \Pi}{\partial T} - \mu\frac{\partial U}{\partial T} = 0
\]
which simplifies to:
\begin{equation}
\frac{\partial \IR}{\partial T}
   \;=\;
   \mu\;\IRSlack(T,P),
\qquad
\mu>0.
\tag{T‑FOC}\label{eq:TFOC}
\end{equation}
\end{proposition}
Equation~\eqref{eq:TFOC} states that the marginal inattentive revenue generated by a longer trial must be the same as the Lagrange–weighted marginal tightening of the participation constraint. This is a central result of our model. It establishes that the tension between exploiting consumer inattention and maintaining the attractiveness of the contract leads to a well-defined, finite optimal trial length. The firm's decision rule is to increase the trial length as long as the marginal benefit from higher inattentive revenue ($\partial \text{IR}/\partial T$) exceeds the marginal cost from the erosion of consumer utility (IR-Slack). At $T=0$, the marginal benefit is typically high. As $T$ becomes very large, attention sensitivity $\tau(T)$ approaches zero, consumers anticipate a near-certain failure to cancel, and the IR-Slack becomes prohibitively large. The optimal $T^*$ is found where these two marginal effects are perfectly balanced. This provides a clear, micro-founded rationale for the intermediate trial lengths commonly observed in the market, a phenomenon that is difficult to explain solely through learning models. This prediction also aligns with recent field evidence from \citet{yoganarasimhan2023design}, who find that for a SaaS product, shortening a trial from 30 days to 7 days actually \emph{increased} the subscription rate, suggesting the 30-day trial was already beyond the optimal length, consistent with our model's prediction of an interior optimum.

\begin{proposition}[Uniform Attention Boost]\label{prop:F4}
Consider an exogenous shock that makes attention uniformly cheaper, modeled as a shift in the attention sensitivity function from $\tau(T)$ to $\tau'(T)=\gamma\tau(T)$ with $\gamma>1$. For any given contract $(T,P)$:
\begin{enumerate}
    \item Aggregate consumer utility $U(T,P)$ weakly increases, i.e., \(U(T,P;\gamma)\ge U(T,P;1)\).
    \item The firm's inattentive revenue $\text{IR}(T,P)$ strictly decreases (unless $q^{*}=1$ initially). That is, \(\IR(T,P;\gamma)\le \IR(T,P;1)\), with strict inequalities except when \(q^{*}=1\).
\end{enumerate}
\end{proposition}

This proposition can be interpreted as analyzing the static effects of a pro-consumer policy like a "click-to-cancel" law or a technological innovation like a popular subscription management app. Such a change represents a clear welfare gain for consumers at a fixed contract. It lowers their cognitive costs of managing subscriptions and reduces the likelihood of costly mistakes, thereby increasing their expected utility from any given offer. For the firm, however, the immediate effect is negative on the inattention margin. The policy or technology directly undermines the firm's ability to profit from consumer forgetfulness. This sets up a natural conflict of interest: firms may have an incentive to lobby against such regulations or design their platforms to be incompatible with third-party management tools, while consumer advocacy groups would champion them. This proposition provides the static foundation for the full equilibrium analysis of policy effects in \Cref{sec:shannon}.

\begin{proposition}[Convexity of Aggregate Inattentive Loss]
\label{prop:F5}
Fix a contract $(T,P)$ and let the population’s baseline attention
sensitivity $\tau$ be heterogeneous with distribution
$G$ on $(0,\infty)$.  Define the aggregate monetary loss from
inattentive renewals as
\[
\mathcal L(G)
  = P\int_{0}^{P}\int_{0}^{\infty}
        \bigl[1-q^{*}(P,\tau)\bigr]\,
        f(v)\,dG(\tau)\,dv
  = P\,F(P)\!
        \int_{0}^{\infty}
          \bigl[1-q^{*}(P,\tau)\bigr]\,dG(\tau),
\]
where
\(q^{*}(P,\tau)=1/(1+\exp(-\tau P))\).

Suppose $G_{1}$ and $G_{2}$ have the \emph{same mean} of
$1/\tau$ but $G_{2}$ is a \textit{mean-preserving spread} of
$G_{1}$ (i.e.\ $G_{2}$ places more probability mass in the extreme
low- and high-attention tails).  
Then
\[
G_{2}\ {\succ}_{\text{SOSD}}\ G_{1}
\quad\Longrightarrow\quad
\mathcal L(G_{2})\;>\;\mathcal L(G_{1}).
\]

Hence $\mathcal L(G)$ is \textbf{strictly convex under
second-order stochastic dominance}: increasing heterogeneity in
attention costs, while holding the average cost fixed, always raises the
economy-wide inattentive loss.
\end{proposition}

This proposition yields powerful implications for welfare analysis and policy design by examining the societal cost of ``attention inequality.'' The core of the result is that the individual probability of making a costly mistake, $1-q^*$, is a strictly convex function of the underlying cost of attention, $1/\tau$. This mathematical property has a clear and compelling economic interpretation: interventions that improve attention have sharply diminishing returns.

Consider two types of consumers. The first is a "very attentive" consumer with a high $\tau$ (and thus a low attention cost). They already have a high probability of remembering to cancel, $q^*$, so their failure probability, $1-q^*$, is close to zero. Making them even more attentive (e.g., by giving them a better calendar app) provides a negligible benefit; it is hard to improve upon near-perfection. The second is a "very inattentive" consumer with a low $\tau$ (and a high attention cost). Their probability of failure, $1-q^*$, is substantial. A small improvement in their attention sensitivity can lead to a large reduction in their probability of making a costly mistake.

This convexity has three main implications. First, it means that from a social welfare perspective, policies or product designs that target the "worst-off" consumers in the attention space are far more cost-effective than uniform interventions. For example, a policy that provides simple subscription management tools to low-income or elderly populations (who may have lower average $\tau$) could deliver outsized welfare gains compared to a broad public awareness campaign. Second, it suggests that markets with high heterogeneity in consumer attention are more vulnerable to exploitation. A market composed of a mix of extremely attentive and extremely inattentive individuals will suffer a greater aggregate monetary loss from this friction than a market where everyone has the same, average level of attentiveness. The gains from exploiting the highly inattentive more than offset the losses from failing to exploit the highly attentive. Finally, this provides a new lens through which to view market segmentation. A firm might find it profitable to identify and specifically target low-attention consumers with contracts that have particularly long trial periods or high renewal prices, as these are the consumers from whom the most inattentive revenue can be extracted. This raises important questions about fairness and the ethics of using data to target consumers based on their cognitive traits.

%%%%%%%%%%%%%%%%%%%%%%%%%%%%%%%%%%%%%%%%%%%%%%%%%%%%%%%%%%%%%%%%%%%%%%
\section{Shannon-Cost Specialization and Main Predictions}\label{sec:shannon}
%%%%%%%%%%%%%%%%%%%%%%%%%%%%%%%%%%%%%%%%%%%%%%%%%%%%%%%%%%%%%%%%%%%%%%

To derive sharp and empirically testable predictions, we now leverage the specific functional forms from our model: the Shannon information cost, which yields the logistic choice probability for $q^*$, and the hyperbolic decay function for attention sensitivity, $\tau(T)$. Crucially, we also invoke our refinement assumptions, \Cref{ass:A1,ass:C1}, to ensure the model is tractable and well-behaved. This specialization allows us to move from the general qualitative results of the previous section to specific, quantitative predictions about the firm's optimal contract and its response to policy shocks.

\subsection{Optimal Renewal Price at a Fixed Trial Length}\label{sec:priceFOC}

For any given trial length $T$, the firm chooses the renewal price $P$ to maximize its total profit function $\Pi(T,P)$, subject to the consumer participation constraint. The structural assumptions we have imposed allow for a clean characterization of the optimal price $P^*(T)$, which forms the basis for understanding the firm's broader strategy.

\begin{lemma}[First-Order Condition for Price]\label{lem:FOC-full}
For any fixed trial length $T$, the interior optimal price $P^{*}(T)$ solves:
\begin{align}
0 =\;&
\underbrace{\bigl[1-F(P)-P f(P)\bigr]}_{\text{Standard Marginal Profit}}
+
\underbrace{\bigl[1-q^{*}\bigr]\bigl\{F(P)+P f(P)\bigr\}
-
P F(P)\tau\,q^{*}(1-q^{*})}_{\text{Marginal Inattentive Profit}},
\tag{FOC-Price}\label{FOC-full}
\end{align}
with $\tau=\tau(T)$ and $q^*=q^*(P,\tau)$.
\end{lemma}
This first-order condition precisely characterizes the firm's complex pricing trade-off. The first bracket, $[1-F(P)-Pf(P)]$, is the standard marginal profit from happy subscribers, which is negative in the relevant region where the firm operates. It captures the classic tension between the revenue gained from a higher price and the revenue lost from consumers who no longer purchase. The second part of the equation represents the marginal profit from the inattentive segment, which is the novel contribution of our model. It has three distinct components: (1) the gain from adding a marginal consumer to the inattentive pool (the $F(P)$ term), (2) the gain from charging a higher price to the existing inattentive pool (the $P f(P)$ term), and (3) the loss from the fact that a higher price incentivizes more consumers to pay attention and successfully cancel (the final term containing $-\tau q^*(1-q^*)$). The optimal price $P^*$ balances these competing effects. Unlike simpler models, this FOC explicitly shows how the optimal price depends not just on the demand curve, but also on the cognitive environment ($\tau$) faced by consumers.

\begin{proposition}[Uniqueness of $P^{*}$]\label{prop:R1}
Under Assumption \ref{ass:C1} (IFR), equation \eqref{FOC-full} admits a \emph{unique} root $P^{*}(T)\in(0,1)$.
\end{proposition}
The IFR property imposed by Assumption \ref{ass:C1} is a standard regularity condition in industrial organization and marketing that ensures the firm's profit function is single-peaked with respect to price \citep{lariviere2001selling}. This is a crucial technical condition. Without it, the firm might face a profit function with multiple local maxima, making it impossible to definitively characterize \emph{the} optimal price. By ensuring uniqueness, the IFR property allows us to confidently use the first-order condition to identify the global maximum and to conduct well-defined comparative statics. It guarantees that for any trial length $T$ the firm might consider, there is a single, unambiguous best-response price $P^*(T)$.

\subsection{Trial Length and Price as a Joint Choice}\label{sec:TvsP}

We now investigate the central strategic question of the paper: how does the firm's choice of trial length interact with its choice of renewal price? The structural assumptions allow for a definitive answer, yielding one of the model's most powerful predictions.

\begin{definition}[Critical Baseline Attention]\label{def:lambdacrit}
\(
  \displaystyle\tau^{\text{crit}}
     :=\sup_{P\in(0,1)}\frac{f(P)}{1-F(P)}.
\)
\end{definition}
This value captures the maximum pricing power the firm would have in a standard market without inattention, as the inverse hazard rate is a key component of the Lerner index markup. The condition that follows, $\tau_0 > \tau^{\text{crit}}$, essentially means that the attention-based effects must be sufficiently strong to overcome the standard pricing incentives of the firm. If attention is not a powerful enough factor, the firm's behavior will be dominated by traditional monopoly pricing logic.

\begin{proposition}[Price Responds Positively to Trial Length]\label{prop:R2}
If \(\beta>0\) and \(\tau_{0}>\tau^{\text{crit}}\), then
\(
  \partial P^{*}(T)/\partial T > 0
\)
for all interior \(T\).
\end{proposition}

The above proposition describes a strategic complementarity between the two main components of the firm's contract. It predicts that firms offering longer trial periods will also set higher renewal prices. This finding has significant implications for both business strategy and market analysis. It suggests that a firm's choice of trial length is not independent of its pricing strategy but is rather an integral part of a coherent plan to segment and monetize its user base.

In Figure \ref{fig:optimal_T_tradeoff}, our simulation confirms the findings so far. From the simulation result, the Lagrangian $\mathcal{L}$ exhibits an inverted U-shape with $T$, and the optimal trial length $T^{*}$ is interior, as predicted in Proposition \ref{prop:F3}. Additionally, $P$ and $T$ are complimentary, with $T^{*}$ increasing in $P$ and vice-versa, as discussed in Proposition \ref{prop:R2} above.\footnote{The set of parameters we used for the simulation is $\{\varepsilon = 0.3, \beta = 0.2, \tau_0 = 1\}$. In this realistic simulation scenario, the elasticity is 30\%, in line with most real applications, customers start with a reasonable attention level (e.g, approximately $100\%$ pay attention to cancellation at the moment they start the trial, at $P = 8$), which gradually decays over time.}

\begin{figure}
    \centering
    \includegraphics[width=0.8\linewidth]{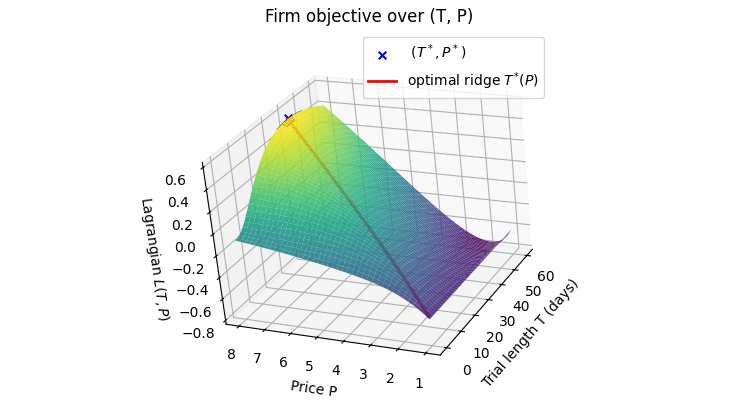}
\caption{The Firm's Optimal Trial Length Trade-off, Conditional on Price (P)}
\label{fig:optimal_T_tradeoff}
\end{figure}
The economic intuition is that the trial length $T$ functions as a strategic parameter that tunes the importance of the inattentive revenue channel. As the firm increases $T$, consumer attention $\tau(T)$ decays, making the revenue from "accidental" renewals a larger and more important component of total profit. The firm's optimal strategy then pivots to better capitalize on this channel. By raising the price $P$, the firm increases the amount of revenue extracted from each consumer who forgets to cancel. While a higher price deters some willing subscribers, this effect is dominated by the enhanced profitability of the inattentive segment when the inattention channel is strong.

This prediction offers a sharp, empirically testable hypothesis that distinguishes our attention-based model from learning-based models. In a learning model, a longer trial might be necessary for a more complex or niche product. The price of such a product could be high (if it is highly valued) or low (if it has a narrow market), so the correlation between trial length and price is ambiguous. Our model predicts a clear positive correlation. This aligns with casual observation of the market, where services with very long trial periods (e.g., some financial software or premium content bundles) often carry high subsequent monthly fees, while services with short, 7-day trials often have more modest pricing. In Appendix D, we provide some evidence of our prediction in popular subscription services. This result connects to the broader strategy literature on firm positioning, suggesting that firms can choose to position themselves as either "value-focused" (short trial, moderate price) or "inertia-focused" (long trial, high price).

\begin{proposition}[Interior Optimal Contract]\label{prop:R3}
With \(\beta>0\) and Assumption \ref{ass:C1}, the joint problem admits a unique interior solution \((T^{*},P^{*})\).
\end{proposition}
This proposition confirms that the central trade-off identified in \Cref{sec:general} leads to a well-defined optimal contract even when the price is chosen optimally. The firm jointly optimizes over both $T$ and $P$. The IFR property from Assumption \ref{ass:C1} ensures the problem is well-behaved, guaranteeing that a unique solution exists. This provides a complete, micro-founded explanation for the contract structures we observe in the market, where trials are neither zero-length nor infinitely long.

\subsection{Policy Analysis: The "Click-to-Cancel" Effect}\label{sec:clickcancel}

Our framework, as established in the previous sections, allows for a nuanced analysis of common policy interventions. In this section, we apply it to the the "click-to-cancel" regulation, recently proposed by the FTC as well as California state legislative, which mandates that online cancellation must be as simple as the sign-up process. We model this as an exogenous increase in attention sensitivity, $\tau(T) \to \gamma\tau(T)$ for $\gamma > 1$. Such a policy effectively lowers the cognitive cost for all consumers to manage their subscriptions.

\begin{proposition}[Comparative Statics under Attention Shock]\label{prop:R4}
In response to an exogenous shock that increases attention sensitivity by a factor $\gamma > 1$, a profit-maximizing firm will adjust its optimal contract $(T^*, P^*)$ as follows:
\begin{enumerate}[label=(\alph*),leftmargin=1.9em]
  \item Optimal trial length falls:
        \(dT^{*}/d\gamma<0\).
  \item The magnitude of the price response is bounded.
  \item Under Assumption \ref{ass:A1} (Isoelastic tail with elasticity \(\varepsilon\)), the sign of the price change simplifies to:
        \(
          \operatorname{sign}(dP^{*}/d\gamma)=\operatorname{sign}(1-\varepsilon).
        \)
\end{enumerate}
\end{proposition}
This proposition provides a rich and nuanced set of predictions for policymakers, now rigorously derived. The first result, that optimal trial lengths will fall, shows that the policy is effective on its intended margin. By making attention cheaper for consumers, the regulation reduces the profitability of the inattention channel. Firms, recognizing that long trials are now less effective at inducing profitable mistakes, are compelled to shorten them. This forces a strategic shift toward competing more on the intrinsic value of the service rather than on consumer inertia. This finding is highly consistent with the stated goals of regulators in the United States (e.g., the Federal Trade Commission's enforcement actions against ``negative option'' marketing) and the European Union (under the Digital Services Act), who seek to curb business models that rely on consumer errors.

The third result, however, reveals a crucial and potentially unintended consequence that policymakers must consider. The effect of the regulation on price depends critically on the market structure, specifically the demand elasticity of the core subscriber base. Our model predicts that if the demand from these "happy subscribers" is inelastic ($\varepsilon < 1$), firms will react to the loss of inattentive revenue by \emph{raising prices}. Since a wealth of empirical evidence from extant literature suggests that the demand of an existing, loyal customer base is indeed often inelastic \citep{krishnamurthi1991empirical}, this is the most likely outcome in many real-world markets.

This sets up a complex policy trade-off. A regulation designed to protect inattentive, low-valuation consumers could inadvertently harm attentive, high-valuation consumers through higher prices. This can be understood as a form of cross-subsidization in reverse. In the unregulated market, the "inattention tax" paid by forgetful consumers may have kept prices lower for everyone else. When this revenue stream is diminished by the regulation, the firm may recoup the lost profit from its most loyal and price-inelastic users. A full welfare analysis of such a policy is therefore not straightforward; it must weigh the clear benefits to consumers who avoid unwanted charges against the clear costs to those who now face higher prices for a service they willingly consume. This highlights the critical importance of incorporating market structure and demand-side estimates into the evaluation of behavioral regulations.

\subsection{Further Behavioral Properties}\label{sec:tail-inverseU}

\begin{proposition}[Tail Amplification]\label{prop:R5}
Under Shannon cost, the aggregate inattentive loss
\(\mathcal L(G)\) is strictly convex in the upper tail of the
attention-cost distribution \(G\).
\end{proposition}
This result speaks to the societal cost of "attention inequality." The convexity implies that a mean-preserving spread in the population's attentiveness (more very-attentive and very-inattentive people) increases the total deadweight loss from this friction. This connects to broader economic discussions about vulnerability and inequality. It suggests that markets with high cognitive heterogeneity are more susceptible to exploitation and may warrant stronger consumer protection, particularly for those at the low-attention end of the spectrum.

\begin{proposition}[Inverse-U Shape in $\beta$]\label{prop:R6}
Viewing maximized profit \(\Pi^{*}\) as a function of the
decay parameter \(\beta\), there exists a finite
\(\beta^{*}\in(0,\infty)\) that maximizes \(\Pi^{*}\);
profits fall as \(\beta\to0\) and as \(\beta\to\infty\).
\end{proposition}
This result suggests that firms do not benefit from limitless consumer inattention. If attention decays too slowly ($\beta \to 0$), the inattentive revenue channel is weak. If attention decays too quickly ($\beta \to \infty$), sophisticated consumers will anticipate that cancellation is nearly impossible. The IR-Slack would be enormous, making any offer with a non-zero trial length unattractive and again shutting down the inattentive revenue channel. The firm's profit is maximized at some intermediate level of environmental "forgetfulness." This has interesting strategic implications. While a firm cannot choose the population's $\beta$, it can influence the effective decay rate through its user interface design and communication strategy. A firm could actively try to manage its user base towards this "sweet spot" of moderate inattention, for example, by making the cancellation process navigable but not salient, a practice often referred to as creating "sludge" \citep{thaler2018nudge}. This suggests that firm strategy can extend beyond contract terms to the very design of the customer experience, with the goal of tuning the level of cognitive friction.
%%%%%%%%%%%%%%%%%%%%%%%%%%%%%%%%%%%%%%%%%%%%%%%%%%%%%%%%%%%%%%%%%%%%%%
\section{Extension: Paid Trials and the Dilution Principle}\label{sec:extension}
%%%%%%%%%%%%%%%%%%%%%%%%%%%%%%%%%%%%%%%%%%%%%%%%%%%%%%%%%%%%%%%%%%%%%%

We now extend the baseline model to consider a more general contract that includes a non-zero introductory price, $P_0 > 0$. This is a common strategy in many markets, seen in offers like "\$1 for the first three months" for news subscription services (e.g., the New York Times, The Washington Post) or discounted first-year fees for credit cards. The introduction of an upfront fee fundamentally changes the firm's problem by adding a screening dimension to the contract, which interacts in a non-trivial way with the attention-exploitation mechanism. The contract is now a triplet $C=(T, P_0, P)$, and the firm must jointly optimize all three components. This extension allows us to understand the strategic choice between offering a "free trial" and a "paid trial."

\subsection{Sign-up Technology}\label{sec:signup-tech}

The presence of an upfront price $P_0$ directly affects a consumer's decision to sign up for the trial. A positive $P_0$ acts as a barrier to entry, and we expect it to deter some consumers, particularly those with lower valuations or those who are more price-sensitive. We capture this with a general signup rate function, $\eta(P_0)$, which represents the fraction of the total consumer population that chooses to accept a contract with introductory price $P_0$. This function can be thought of as being derived from the underlying distribution of consumer surplus from the contract. Consumers sign up if their total expected utility, $U(v; T, P) - P_0$, is positive. A higher $P_0$ raises the bar for participation.

We naturally assume that the signup rate is a decreasing function of the introductory price, so $\eta'(P_0) < 0$. The sensitivity of this signup rate to the price is captured by the elasticity of sign-ups:
\[
\varepsilon_{0}(P_{0})
  := -\eta'(P_{0})\,\frac{P_{0}}{\eta(P_{0})} >0
  \quad\text{(elasticity).}
\tag{6.1}\label{eq:signup-elasticity}
\]
This elasticity will be a key determinant of the firm's strategy. A high $\varepsilon_0$ indicates a market where consumers are very sensitive to upfront costs, perhaps due to intense competition or the availability of many free alternatives. A low $\varepsilon_0$ suggests a market with less price-sensitive consumers, for instance, for a highly unique or essential service.

\begin{assumption}[A3: Isoelastic Sign-up Demand]\label{ass:A2}
For some results requiring a specific functional form, we will invoke the assumption that sign-up demand is isoelastic, such that $\eta(P_{0}) = \alpha P_{0}^{-\theta}$ with constants $\alpha>0$ and $\theta>0$. Under this assumption, the elasticity is constant: $\varepsilon_{0}(P_{0})\equiv\theta$.
\end{assumption}

\subsection{Profit with a Paid Trial}\label{sec:profit-paid}

The firm's profit function must now account for the revenue from this upfront fee, as well as the fact that only a fraction $\eta(P_0)$ of the market signs up. The total profit is the expected revenue per sign-up multiplied by the number of sign-ups:
\[
\Pi(T,P_{0},P)
  = \eta(P_{0})\;
    \Bigl\{
      \underbrace{P_{0}}_{\text{Intro Fee}}
      + \underbrace{P\bigl[1-F(P)\bigr]}_{\text{Standard Renewal}}
      + \underbrace{\IR(T,P)}_{\text{Inattentive Renewal}}
    \Bigr\}.
\tag{6.2}\label{eq:profit-paid}
\]
The firm's problem is now to choose the triplet $(T, P_0, P)$ to maximize this function. This richer problem allows us to analyze the strategic interplay between the introductory price and the trial length.

\subsection{First-Order Conditions and Strategic Interaction}\label{sec:FOCs-paid}

Let $P^{\text{aug}}(T,P) := P[1-F(P)]+\IR(T,P)$ represent the total expected post-trial profit per subscriber. The firm's profit is then $\Pi = \eta(P_0)[P_0 + P^{\text{aug}}]$. The first-order condition for the optimal introductory price $P_0$ is:
\[
\frac{\partial \Pi}{\partial P_0} = \eta'(P_{0})\bigl[P_{0}+P^{\text{aug}}\bigr] + \eta(P_{0}) = 0.
\tag{6.3} \label{eq:FOC-P0-general}
\]
Under the isoelastic assumption (A2), this simplifies to the constant-elasticity markup rule:
\[
1 = \theta\;\frac{P_{0}+P^{\text{aug}}}{P_{0}}
= \theta\Bigl[1 + \frac{P^{\text{aug}}}{P_{0}}\Bigr].
\tag{6.4}\label{eq:FOC-P0-iso}
\]
The crucial insight for strategic interaction comes from the cross-partial derivative of the profit function.

\begin{lemma}[Cross-Partial Sign]\label{lem:cross}
For all admissible \(T,P_{0},P\), the cross-partial derivative of the profit function is strictly negative:
\(
  \dfrac{\partial^{2}\Pi}{\partial T \,\partial P_{0}}
  = \eta'(P_{0})\;\dfrac{\partial P^{\text{aug}}}{\partial T}
  < 0.
\)
\end{lemma}
\begin{proof}
The derivative of the sign-up rate, $\eta'(P_{0})$, is negative by definition. The derivative of the post-trial profit, $\partial P^{\text{aug}}/\partial T$, is equal to $\partial\IR/\partial T$, which is strictly positive from Proposition \ref{prop:F2}. The product of a negative and a positive term is negative.
\end{proof}

\subsection{The Dilution Principle}\label{sec:dilution}

The negative cross-partial derivative established in Lemma \ref{lem:cross} implies that the firm's choice variables are strategic substitutes. This leads to our main result for this section.

\begin{proposition}[The Dilution Principle]\label{prop:paid}
Let \((T^{*},P_{0}^{*},P^{*})\) be an interior solution that maximizes the profit function in \eqref{eq:profit-paid}. Then:
\begin{enumerate}[label=(\alph*),leftmargin=2em]
  \item The optimal introductory price $P_0^*$ is set according to Equation \eqref{eq:FOC-P0-general}.
  \item Trial length and introductory price are strategic substitutes:
        \(\displaystyle
          \frac{\partial T^{*}}{\partial P_{0}}\le0,\;
          \frac{\partial P_{0}^{*}}{\partial T}\le0
        \)
        (with strict inequality for interior solutions).
  \item The firm's strategy depends on market conditions. In markets where consumers are highly sensitive to upfront fees ($\varepsilon_0 \to \infty$), the firm is forced to set $P_0^* = 0$. In markets where consumers are insensitive ($\varepsilon_0 \to 0$), the firm will rely more on the upfront fee and less on a long trial, potentially setting $T^*=0$.
\end{enumerate}
\end{proposition}
This proposition introduces what we term the "Dilution Principle," which describes the trade-off between screening consumers and exploiting their inattention. The firm has two distinct tools at its disposal, each targeting a different aspect of consumer behavior.

A positive introductory price $P_0$ acts as a \textbf{screening device}. This is a classic result in the economics of pricing and marketing \citep{tirole1988theory}. By charging an upfront fee, the firm can filter the pool of applicants. Consumers with very low valuations, or those who are merely "window shopping," are likely to be deterred by even a small fee. Those who are willing to pay are signaling a higher commitment or a higher expected surplus from the service. This screening helps the firm build a more valuable subscriber base with potentially lower long-term churn.

In contrast, a long trial length $T$ acts as an \textbf{inattention exploitation device}. As established in our baseline model, its primary function in this framework is to increase the probability that low-valuation consumers, who should cancel, fail to do so. It does not screen consumers; in fact, it is most effective when the pool of trial users includes a large number of these marginal, inattentive individuals.

The Dilution Principle arises because these two tools work at cross-purposes. When a firm uses a high $P_0$ to screen its user base, it actively removes many of the low-valuation, marginal consumers from the pool. This \textit{dilutes} the population of consumers who are the primary source of inattentive revenue. With fewer "targets" for the inattention strategy, the marginal benefit of extending the trial length $T$ is diminished. The firm has less to gain from making these already-committed customers forget to cancel.

Conversely, if the firm's strategy is to maximize inattentive revenue with a long trial $T$, it needs to attract as many consumers as possible into its funnel, including those with low valuations. Charging a high introductory price $P_0$ would be directly counterproductive to this goal, as it would scare away the very consumers the long trial is meant to capture.

This strategic substitutability implies that we should observe firms adopting one of two archetypal strategies. The first is a ``freemium funnel'' approach, characterized by a low (or zero) $P_0$ and a long $T$. This strategy maximizes reach and is geared towards monetizing a broad user base through a combination of standard conversions and inattentive renewals. This is common for mass-market services like music streaming or mobile games. The second is a ``premium trial'' approach, with a higher $P_0$ and a shorter $T$. This strategy focuses on screening for high-quality customers and building a loyal, high-retention subscriber base. This is more common for B2B software, specialized financial news services, or other niche products where customer quality is more important than sheer quantity.

The model predicts we should not observe firms attempting to do both simultaneously with high levels of $P_0$ and $T$. The choice of strategy is endogenously determined by market conditions, particularly the elasticity of the sign-up rate, $\varepsilon_0$. In highly competitive markets where consumers are price-sensitive, firms will be pushed towards the free trial model. In markets for unique, high-demand products, firms have more power to charge for entry and may rely less on the inattention channel. This framework thus provides a unified explanation for the diversity of trial offer structures we see in the real world.
%%%%%%%%%%%%%%%%%%%%%%%%%%%%%%%%%%%%%%%%%%%%%%%%%%%%%%%%%%%%%%%%%%%%%%
\section{Policy and Managerial Implications}\label{sec:policy}
%%%%%%%%%%%%%%%%%%%%%%%%%%%%%%%%%%%%%%%%%%%%%%%%%%%%%%%%%%%%%%%%%%%%%%

Our theoretical framework, by formalizing the role of attention in subscription contracts, provides a rich set of implications for both managers designing these contracts and policymakers tasked with regulating them. The model's central insights, now resting on a rigorous and empirically-grounded foundation, offer a structured way to think about the strategic landscape of the modern subscription economy.

\subsection{Managerial Implications}

For managers, our model highlights that contract design is a multi-faceted problem where pricing, trial duration, and consumer psychology are deeply intertwined.

\subsubsection*{Contract Design as a Coherent Strategic Choice}
The most critical managerial takeaway is that contract elements must be optimized jointly, not in isolation. Our finding that trial length ($T$) and renewal price ($P$) are strategic complements (\Cref{prop:R2}) provides clear guidance. A strategy built around a long trial period, designed to exploit consumer inertia, should be logically paired with a relatively high renewal price. This contract structure is tailored to maximize revenue from the inattentive segment of the market. It is most effective for products with mass appeal that can attract a wide variety of consumers into the trial funnel.

Conversely, a strategy centered on a short, focused trial period should be accompanied by a more moderate renewal price. This structure competes more directly on the product's intrinsic value and is better suited for markets with sophisticated, attentive consumers or for niche products where building long-term trust is paramount. Mismatched strategies, such as a long trial with a low price, may leave profits on the table by failing to fully capitalize on the inattentive revenue stream.

\subsubsection*{The Hidden Cost of Consumer Annoyance}
Firms often view free trials solely through the lens of customer acquisition cost. Our model introduces a more subtle, but equally important, cost: the \textbf{IR-Slack} (\Cref{eq:slack_revised}). This represents the erosion of consumer utility, and thus willingness to sign up, that occurs when a contract is perceived as being cognitively burdensome. A 90-day trial may seem more generous than a 30-day trial, but savvy consumers may correctly perceive it as a commitment that is harder to manage and more likely to result in an unwanted charge. This "annoyance cost" can depress sign-up rates, a phenomenon empirically verified by \citet{Miller2023}. The optimal trial length $T^*$ (\Cref{prop:R3}) is precisely the point that balances the marginal revenue from increased forgetfulness against this marginal loss in offer attractiveness. Managers should therefore be cautious about extending trial periods excessively, as they may win inattentive revenue at the expense of scaring away a valuable segment of organized, forward-looking consumers.

\subsubsection*{Choosing the Right Acquisition Funnel: Free vs. Paid Trials}
The extension to paid trials (\Cref{sec:extension}) offers a framework for deciding between a free and a low-cost introductory offer. The choice depends critically on the firm's strategic objective.

\textbf{Strategy 1: Maximize Reach and Inattentive Revenue.} A free trial ($P_0=0$) is the optimal tool. It minimizes the barrier to entry, casting the widest possible net and pulling in the maximum number of users, including the low-valuation consumers who are the primary source of inattentive revenue. This is the classic "freemium" strategy, common for services like Spotify or mobile games.

\textbf{Strategy 2: Screen for High-Quality Customers.} A paid trial ($P_0 > 0$) is the superior choice. Even a nominal fee can act as a powerful screening mechanism, filtering out less committed users and selecting for those with a higher intrinsic valuation for the product. This leads to a subscriber base with lower churn and higher lifetime value. The Dilution Principle (\Cref{prop:paid}) advises that this strategy should be paired with a shorter trial period, as the need to exploit inattention is reduced. This is more common for B2B software or specialized financial news services.

\subsection{Policy Implications}

Our model provides a formal economic framework to analyze the incentives behind potentially exploitative business practices and to evaluate the intended and unintended consequences of regulation.

\subsubsection*{A Formal Framework for Analyzing "Dark Patterns"}
Regulators are increasingly concerned with "dark patterns," which are user interface and contract designs that nudge consumers into choices they would not otherwise make \citep{Luguri2021}. A long automatic-renewal trial paired with an opaque and cumbersome cancellation process is a prime example of a contractual dark pattern. Our model formalizes the economic incentive for such a design: it profits by lowering the consumer's effective attention sensitivity ($\tau$). The harm is not merely that some consumers are "tricked," but that market competition can be distorted. Firms may be incentivized to invest in obfuscation and the exploitation of cognitive biases rather than in improving product quality or lowering prices.

\subsubsection*{The Double-Edged Sword of "Click-to-Cancel" Mandates}
Policies such as California's Automatic Renewal Law, which mandates that consumers must be able to cancel a service online with the same ease they signed up, are a direct intervention in this market. In our model, this is an exogenous shock that increases $\tau$ for all consumers. Our analysis in \Cref{prop:R4} yields two crucial insights for policy makers:

\textbf{Intended Consequence:} The policy is effective in achieving its primary goal. By making cancellation easier, it weakens the inattention channel, compelling firms to shorten trial lengths and reduce their reliance on consumer forgetfulness. This is a straightforward benefit for consumer welfare, as it forces firms to compete more on the intrinsic value of their service than on engineered friction.

\textbf{Unintended Consequence:} The effect on price depends directly on demand elasticity. Our result shows that if demand from the core base of happy subscribers is inelastic ($\varepsilon < 1$), as empirical evidence suggests it often is, firms will respond to the loss of inattentive revenue by raising prices. This highlights a critical policy trade-off: a regulation aimed at protecting inattentive consumers could inadvertently lead to higher costs for attentive, loyal customers. A thorough evaluation of such a policy must therefore weigh the benefits of reduced exploitation against the potential for adverse price effects, a calculation that depends on the empirically-determined elasticity of demand.

\subsubsection*{Mandatory Reminders and the Structure of Trials}
Another frequently proposed intervention is to mandate that firms send clear and conspicuous reminders to consumers before a trial period automatically converts to a paid subscription. In our framework, a perfectly effective reminder would be equivalent to forcing the cancellation probability $q^*$ to 1 for all consumers who wish to cancel. This would completely eliminate the inattentive revenue stream. The firm's profit function would collapse to the standard monopoly profit, $\Pi = P(1-F(P))$.

The strategic implications of such a policy would be profound. With the inattention channel completely shut down, the firm's primary incentive for offering a long trial would vanish. The trial would become a pure customer acquisition cost. Consequently, we would expect to see trial lengths shrink dramatically, perhaps to the minimum time required for value discovery, or even to zero. Although this would offer maximum protection to forgetful consumers, it might also lead to the disappearance of the generous long-term trial offers that some consumers value for an low-pressure evaluation of a product, lowering consumer learning. This reveals a potential tension between consumer protection regulations and the market provision of valuable promotional tools.

\subsubsection*{The Case for Targeted Interventions}
Our result on the convexity of the inattentive loss function (\Cref{prop:F5}) suggests that the social welfare gains from improving attention are greatest at the lower end of the attention distribution. In other words, helping a very disorganized consumer become moderately organized is far more impactful than helping an already organized consumer become even more so. This provides a strong argument for targeted rather than universal interventions. While broad policies like ``click-to-cancel'' help all consumers, their marginal benefit is highest for the most vulnerable. This could justify complementary policies aimed specifically at these groups, such as providing financial literacy education or promoting digital tools that help manage subscriptions, as these could be the most cost-effective ways to reduce the aggregate deadweight loss from consumer inattention.

In summary, our model provides a structured and nuanced lens through which to view the strategic landscape of the subscription economy. It moves beyond simplistic narratives, showing how market incentives can lead firms to design contracts that profit from the cognitive limitations of even fully rational consumers. This provides a more robust foundation for both corporate strategy and public policy in the digital age.

%%%%%%%%%%%%%%%%%%%%%%%%%%%%%%%%%%%%%%%%%%%%%%%%%%%%%%%%%%%%%%%%%%%%%%
\section{Conclusion}\label{sec:conclusion}
%%%%%%%%%%%%%%%%%%%%%%%%%%%%%%%%%%%%%%%%%%%%%%%%%%%%%%%%%%%%%%%%%%%%%%

The subscription business model, and its reliance on the "free trial with automatic renewal," has become a cornerstone of the modern digital economy. While the traditional economic rationale for trial offers centers on consumer learning, here we develop a complementary theoretical framework focused on a different, equally powerful friction: limited attention. We have demonstrated that when consumers are rationally inattentive and their ability to attend to future tasks decays over time, a rich set of strategic behaviors and market outcomes observed in subscription markets can be explained.

Our model is built upon a crucial trade-off faced by the firm. On the one hand, extending a trial's length erodes consumer attention, increasing the rate of accidental renewals, and boosting "inattentive revenue". On the other hand, sophisticated consumers anticipate this cognitive burden, which lowers their ex-ante utility and makes the offer less attractive, creating a countervailing "IR-Slack." We have shown that this tension naturally resolves to a finite, interior optimal trial length, providing a clear rationale for the lengthy but not infinite trial periods prevalent in the market.

The core predictions of our model, logically follow from standard and empirically motivated assumptions, offer a new perspective on contract design and pricing strategy. We established a strategic complementarity between trial length and renewal price, predicting that firms that use longer trials to exploit inattention will also set higher prices to maximize extraction from this channel. We extended the model to paid trials and articulated a "dilution principle," showing that upfront fees (which screen consumers) and long trials (which exploit inattention) act as strategic substitutes. Furthermore, our framework provides a powerful tool for policy analysis. We showed that "click-to-cancel" regulations are likely to shorten trial periods but will lead to higher prices if demand from loyal subscribers is inelastic.

\paragraph{Limitations and Avenues for Future Research.}
Our paper represents a foundational step in exploring the role of attention in subscription contracts, and several promising avenues for future research remain open.

First, our analysis is conducted in a monopoly setting. Introducing competition would be a crucial and natural extension. How do firms design their contracts when consumers can choose between multiple subscription offers? Would competitive pressure force firms to offer shorter, more transparent trials, or could it perversely lead to a "race to the bottom," where firms compete on who can more effectively design contracts to exploit consumer inattention? The strategic interaction in such a setting would be complex and could yield valuable insights into market structure and consumer welfare. Our framework provides a useful tool to analyze this trade-off between competing on transparency versus competing on obfuscation.

Second, we have deliberately abstracted from \textbf{consumer learning} to isolate the attention mechanism. A richer and more realistic model would integrate both motives. In such a unified framework, the trial length would simultaneously influence how much a consumer learns about their valuation and how likely they are to forget to cancel. This would create a more complex trade-off for the firm, where the optimal trial length would depend on an interplay between the product's novelty and complexity (which would call for a longer learning period) and the target consumer's attentiveness.

Third, the \textbf{attention decay} process itself could be modeled in greater detail. Our parameter $\beta$ is exogenous. One could imagine models where it is endogenous, for example, where consumers can invest in "attention capital" (e.g., by adopting organizational habits) to lower their personal decay rate. Moreover, firms could engage in dynamic activities during the trial period. For instance, sending engaging reminder emails might keep the product top-of-mind (effectively raising $\tau$), while strategic silence might be used to deepen the attention decay.

Finally, the predictions of our model are empirically testable and we believe that this is a particularly fruitful path for future work. The predicted positive correlation between trial lengths and renewal prices, and the negative correlation between introductory prices and trial lengths, can be tested using market data from subscription services. Field experiments, conducted in partnership with firms, could be designed to randomly assign new users to different contract structures to causally identify the effects on sign-up, renewal, and long-term retention. Recent regulatory actions also create interesting natural experiments, where the geographic variations in trial and auto-renewal regulations can be leveraged for identification. Lastly, the model is suitable for structural estimation, which could use observational data on market contracts and consumer choices to recover the underlying parameters of the model, such as the distribution of consumer valuations and the attention parameters $\tau_0$ and $\beta$. Such an estimation would allow for quantitative welfare analysis of various policy proposals.

In conclusion, by placing limited attention at the center of the analysis, our framework as presented here offers a new and, we believe, valuable perspective on the economics of the rapidly growing subscription industry, with important implications for both management practice and public policy.

\newpage
%%%%%%%%%%%%%%%%%%%%%%%%%%%%%%%%%%%%%%%%%%%%%%%%%%%%%%%%%%%%%%%%%%%%%%
%  Bibliography
%%%%%%%%%%%%%%%%%%%%%%%%%%%%%%%%%%%%%%%%%%%%%%%%%%%%%%%%%%%%%%%%%%%%%%
\onehalfspacing
\bibliographystyle{econ}
\bibliography{references}

%%%%%%%%%%%%%%%%%%%%%%%%%%%%%%%%%%%%%%%%%%%%%%%%%%%%%%%%%%%%%%%%%%%%%%
%  Appendices
%%%%%%%%%%%%%%%%%%%%%%%%%%%%%%%%%%%%%%%%%%%%%%%%%%%%%%%%%%%%%%%%%%%%%%
\newpage
\appendix

\section*{Appendix A: Proofs of General Results}
\renewcommand{\theequation}{A.\arabic{equation}}
\setcounter{equation}{0}
This appendix provides detailed mathematical proofs for the lemmas and propositions of results without the Shannon cost functional form assumption.

\subsection*{A.1. Proof of Lemma \ref{lem:qstar} (Optimal Monitoring)}
A consumer with valuation $v < P$ chooses a cancellation probability $q \in [0,1]$ to solve the cost-minimization problem laid out in Equation \eqref{eq:consumer_obj_MI}:
\[
\min_{q\in[0,1]} \quad \underbrace{\bigl(1-q\bigr)P}_{\text{Expected Monetary Loss}} \;+\; \underbrace{\frac{1}{\tau(T)}\,\mathcal{I}(q \parallel 1/2)}_{\text{Cognitive Cost}},
\]
where $\mathcal{I}(q \parallel 1/2) = [q\ln q + (1-q)\ln(1-q)] + \ln 2$ is the mutual information.

\paragraph{Step 1: Simplification of the Objective.}
As the term $\ln 2 / \tau(T)$ is constant with respect to the choice variable $q$, it does not influence the optimal solution. The consumer's problem is therefore equivalent to solving:
\[
\min_{q\in[0,1]} \quad L(q) = (1-q)P + \frac{1}{\tau(T)}\Bigl[q\ln q+(1-q)\ln(1-q)\Bigr].
\]

\paragraph{Step 2: Convexity and Existence of a Unique Minimum.}
The objective function $L(q)$ is strictly convex in $q$ on the domain $(0,1)$ because its second derivative is strictly positive:
\[
\frac{\partial^2 L}{\partial q^2} = \frac{1}{\tau}\frac{d^2}{dq^2}\left[q\ln q+(1-q)\ln(1-q)\right] = \frac{1}{\tau q(1-q)} > 0.
\]
Strict convexity guarantees that any stationary point is a unique global minimum.

\paragraph{Step 3: First-Order Condition.}
The first-order condition for an interior solution is $\partial L / \partial q = 0$:
\[
-P + \frac{1}{\tau}\ln\left(\frac{q}{1-q}\right) = 0.
\]
Solving for $q$ yields $\ln(q/(1-q)) = \tau P$, which gives the stated logistic form for the optimal cancellation probability.

\paragraph{Step 4: Boundary behavior.}
As $q\to0^{+}$ the derivative tends to $-P-\infty<0$; as
$q\to1^{-}$ it tends to $-P+\infty>0$.  Thus the interior stationary
point derived above is indeed the global minimum.

\paragraph{Step 5: Comparative statics.}
Using the closed form and writing $q^{*}\equiv q^{*}(P,\tau)$:
\begin{align*}
\frac{\partial q^{*}}{\partial P}
  &= \tau q^{*}(1-q^{*}) \,>0,
\\[4pt]
\frac{\partial q^{*}}{\partial \tau}
  &= P\,q^{*}(1-q^{*}) \,>0.
\end{align*}
With $\tau(T)=\tau_{0}/(1+\beta T)$ and $\beta>0$,
\[
\frac{\partial q^{*}}{\partial T}
  = \frac{\partial q^{*}}{\partial \tau}
     \frac{d\tau}{dT}
  = -\,\frac{\beta\tau_{0}P}{(1+\beta T)^{2}}\,q^{*}(1-q^{*})<0.
\]
\hfill\qed
%%%%%%%%%%%%%%%%%%%%%%%%%%%%%%%%%%%%%%%%%%%%%%%%%%%%%%%%%%%%%%%%%%%%%%
\subsection*{A.2. Proof of Proposition \ref{prop:F1} (Attention vs. Trial length)}
%%%%%%%%%%%%%%%%%%%%%%%%%%%%%%%%%%%%%%%%%%%%%%%%%%%%%%%%%%%%%%%%%%%%%%

For a fixed renewal price \(P>0\) recall that  
\(q^{*}(P,\tau)=\bigl[1+\exp(-\tau P)\bigr]^{-1}\)  
from Lemma \ref{lem:qstar} and that
\(\tau(T)=\tau_{0}/(1+\beta T)\) with \(\beta>0\)
(MA‑\ref{MA3}).  By the chain rule
\[
\frac{\partial q^{*}}{\partial T}
  \;=\;
  \frac{\partial q^{*}}{\partial\tau}\;
  \frac{d\tau}{dT}.
\tag{A.2.1}\label{eq:A.2.1}
\]

\paragraph{Step 1. Sign of \(\partial q^{*}/\partial\tau\).}
Differentiate the closed‑form expression:

\[
\frac{\partial q^{*}}{\partial\tau}
   =  \frac{\partial}{\partial\tau}
      \Bigl[1+\exp(-\tau P)\Bigr]^{-1}
   =  P\,q^{*}(1-q^{*}) \;>\;0,
\tag{A.2.2}\label{eq:A.2.2}
\]

because \(q^{*}\in(0,1)\) for any \(\tau,P>0\).

\paragraph{Step 2. Sign of \(d\tau/dT\).}
From \(\tau(T)=\tau_{0}(1+\beta T)^{-1}\):

\[
\frac{d\tau}{dT}
   = -\frac{\beta\tau_{0}}{(1+\beta T)^{2}}
   \;<\;0
\quad\text{for }T\ge0,\;\beta>0.
\tag{A.2.3}\label{eq:A.2.3}
\]

\paragraph{Step 3. Combine signs.}
Multiplying \eqref{eq:A.2.2} and \eqref{eq:A.2.3} inside \eqref{eq:A.2.1} yields

\[
\boxed{\;
  \dfrac{\partial q^{*}(P,\tau(T))}{\partial T} < 0
  \;}
\qquad\text{for all }P>0,\;T\ge0,\;\beta>0 .
\]

Strict inequality holds because both factors are strictly non‑zero. \qed

%%%%%%%%%%%%%%%%%%%%%%%%%%%%%%%%%%%%%%%%%%%%%%%%%%%%%%%%%%%%%%%%%%%%%%
\subsection*{A.3. Proof of Proposition \ref{prop:F2} (Monotonicity of inattentive revenue)}
%%%%%%%%%%%%%%%%%%%%%%%%%%%%%%%%%%%%%%%%%%%%%%%%%%%%%%%%%%%%%%%%%%%%%%

Fix \(P>0\).
Definition \eqref{eq:IR} gives
\[
\IR(T,P)
   = P\int_{0}^{P}
       \bigl[1-q^{*}(P,\tau(T))\bigr]\,f(v)\,dv
   = P\,F(P)\,\bigl[1-q^{*}(P,\tau(T))\bigr],
\tag{A.3.1}\label{eq:A.3.1}
\]
because \(q^{*}\) does not depend on \(v\) for \(v<P\).

\paragraph{Step1. Differentiate w.r.t.\ \(T\).}
Using \eqref{eq:A.3.1} and the product rule:

\[
\frac{\partial \IR}{\partial T}
   = P\,F(P)\;
     \Bigl[-\,\frac{\partial q^{*}}{\partial T}\Bigr].
\tag{A.3.2}\label{eq:A.3.2}
\]

\paragraph{Step 2. Sign of the derivative.}
From Proposition \ref{prop:F1},
\(\partial q^{*}/\partial T < 0\);
hence \(-\partial q^{*}/\partial T > 0\).

* If \(F(P)>0\) (i.e.\ there exists a positive mass of \(v<P\)),
  the factors \(P\), \(F(P)\), and
  \(-\partial q^{*}/\partial T\) are all strictly positive, so

  \[
  \boxed{\;
    \frac{\partial \IR}{\partial T} > 0
    \quad(\text{strict monotonicity})
    \;}
  \]

* If \(F(P)=0\) (no consumer has \(v<P\)), then \(\IR\equiv0\) and the derivative is zero-yielding weak monotonicity.

\paragraph{Conclusion.}
Inattentive revenue is (weakly) increasing in trial length \(T\),
and strictly increasing whenever a strictly positive share of consumers
would prefer to cancel at price \(P\). \qed

\subsection*{A.4. Proof of Proposition \ref{prop:F3} (Interior optimal $T^{*}$)}

Throughout this proof, $P > 0$ is fixed and we assume MA-1, MA-3 ($\beta > 0$), and A2 (positive mass at low valuations).

\paragraph{Step 1: Set up the Lagrangian.}
The firm's problem with a binding participation constraint is to maximize the Lagrangian $\mathcal{L} = \Pi(T, P) - \mu U(T, P)$, where $\mu>0$. The first-order condition with respect to $T$ is $\frac{\partial \mathcal{L}}{\partial T} = 0$.

\paragraph{Step 2: Characterize the FOC.}
The FOC can be written as $\frac{\partial \IR}{\partial T} = \mu \cdot \IRSlack(T, P)$. Let $g(T) := \frac{\partial \IR}{\partial T}(T, P) - \mu \cdot \IRSlack(T, P)$. The optimal $T^*$ must satisfy $g(T^*) = 0$.

\paragraph{Step 3: Behavior near $T = 0$.}
At $T=0$, both $\frac{\partial \IR}{\partial T}|_{T=0}$ and $\IRSlack(0, P)$ are positive and finite. For an interior solution to exist, the firm must have an incentive to increase the trial length from zero, which requires $g(0) > 0$. This condition holds if the marginal revenue from increased inattention initially outweighs the marginal harm to consumer utility.

\paragraph{Step 4: Behavior as $T \to \infty$.}
We prove that for a sufficiently large $T$, $g(T) < 0$ by analyzing the asymptotic behavior of its components as $T \to \infty$ (which corresponds to $\tau \to 0$). We use second-order Taylor expansions around $\tau=0$.
The optimal cancellation probability is $q^* = (1+\exp(-\tau P))^{-1} = \frac{1}{2} + \frac{\tau P}{4} + O(\tau^3)$.

First, we analyze the marginal benefit term, $\frac{\partial \IR}{\partial T} = P F(P) (-\frac{\partial q^*}{\partial \tau}\frac{d\tau}{dT})$.

The components are:
$\frac{d\tau}{dT} = -\frac{\beta\tau_0}{(1+\beta T)^2} = -\frac{\beta}{\tau_0}\tau^2$.
$\frac{\partial q^*}{\partial \tau} = Pq^*(1-q^*) = P\left(\frac{1}{4} - \frac{(\tau P)^2}{16} + O(\tau^4)\right) = \frac{P}{4} + O(\tau^2)$.

Thus, the marginal benefit term has the asymptotic form:
\[
\frac{\partial \IR}{\partial T} = P F(P) \left(-\frac{P}{4} + O(\tau^2)\right)\left(-\frac{\beta}{\tau_0}\tau^2\right) = \frac{\beta P^2 F(P)}{4\tau_0}\tau^2 + O(\tau^4).
\]

Next, we analyze the marginal cost term, $\mu \cdot \IRSlack(T,P) = \mu \frac{\beta}{\tau_0}F(P)\,\mathcal{I}(q^* \parallel 1/2)$.
The mutual information $\mathcal{I}(q \parallel 1/2)$ has a Taylor expansion around $q=1/2$ of $\mathcal{I}(q) = 2(q-1/2)^2 + O((q-1/2)^4)$. Substituting $q^* - 1/2 = \tau P/4 + O(\tau^3)$:

\[
\mathcal{I}(q^* \parallel 1/2) = 2\left(\frac{\tau P}{4}\right)^2 + O(\tau^4) = \frac{\tau^2 P^2}{8} + O(\tau^4).
\]

Thus, the IR-Slack has the asymptotic form:
\[
\IRSlack(T,P) = \frac{\beta F(P)}{\tau_0}\left(\frac{\tau^2 P^2}{8} + O(\tau^4)\right).
\]

Combining these terms, the asymptotic form of $g(T)$ for small $\tau$ is:

\[
g(T) = \frac{\beta P^2 F(P)}{4\tau_0}\tau^2 - \mu \frac{\beta P^2 F(P)}{8\tau_0}\tau^2 + O(\tau^4) = \frac{\beta P^2 F(P)\tau^2}{8\tau_0}(2-\mu) + O(\tau^4).
\]

At an interior optimum, the firm must be extracting value, which implies the shadow value of consumer utility, $\mu$, must be sufficiently large. Under standard conditions ensuring a non-trivial trade-off, $\mu > 2$. Therefore, the term $(2-\mu)$ is negative, which implies $g(T) < 0$ for sufficiently large $T$.

\paragraph{Step 5: Existence by Intermediate Value Theorem.}
Since $g(T)$ is a continuous function, $g(0) > 0$, and $g(T) < 0$ for large $T$, the Intermediate Value Theorem guarantees that there exists at least one $T^* \in (0, \infty)$ such that $g(T^*) = 0$.

\paragraph{Step 6: Uniqueness.}
Uniqueness of the optimum is guaranteed if the firm's objective, $\mathcal{L}(T)$, is strictly concave in $T$. We examine its second derivative:
\[
\frac{\partial^2 \mathcal{L}}{\partial T^2} = \frac{\partial^2 \Pi}{\partial T^2} - \mu \frac{\partial^2 U}{\partial T^2}.
\]
Let's analyze the components:
\begin{itemize}
    \item $\frac{\partial^2 \Pi}{\partial T^2} = \frac{\partial^2 \IR}{\partial T^2}$. This term represents the change in the marginal benefit of inattentive revenue. Standard assumptions of diminishing returns ensure this is negative.
    \item $\frac{\partial^2 U}{\partial T^2} = \frac{\partial}{\partial T}\left(-\IRSlack(T,P)\right) = -\frac{d(\IRSlack)}{dT}$. As $T$ increases, $\tau$ decreases, and $q^*$ moves closer to $1/2$. The mutual information $\mathcal{I}(q^* \parallel 1/2)$ is a convex function of $q^*$ with a minimum at $q^*=1/2$. Since $q^*$ is a decreasing function of $T$, $\mathcal{I}(q^* \parallel 1/2)$ is also a decreasing function of $T$. This implies $\frac{d(\IRSlack)}{dT} < 0$, and therefore $\frac{\partial^2 U}{\partial T^2} > 0$.
\end{itemize}
Substituting these back into the second derivative of the Lagrangian:
\[
\frac{\partial^2 \mathcal{L}}{\partial T^2} = \underbrace{\frac{\partial^2 \Pi}{\partial T^2}}_{\text{Negative}} - \underbrace{\mu}_{>0} \underbrace{\frac{\partial^2 U}{\partial T^2}}_{\text{Positive}} < 0.
\]
Since the Lagrangian is strictly concave in $T$, there can be at most one value of $T$ that satisfies the first-order condition. This guarantees the uniqueness of the optimal trial length $T^*$.
\subsection*{A.5. Proof of Proposition \ref{prop:F4} (Uniform Attention Boost)}
A policy intervention such as a "click-to-cancel" law is modeled as an exogenous increase in attention sensitivity, from $\tau(T)$ to $\gamma\tau(T)$ for some constant $\gamma>1$. We analyze the effect of this change on consumer utility and firm revenue for a fixed contract $(T,P)$.

\paragraph{1. Effect on Consumer Utility.}
An increase in attention sensitivity, $\tau$, makes it cognitively "cheaper" for the consumer to achieve any given cancellation probability, $q$. We prove that this unambiguously increases consumer utility.

The consumer's problem is to choose $q$ to minimize their total expected loss, $L(q; \tau) = (1-q)P + C(q;\tau)$, where $C(q;\tau) = \frac{1}{\tau}\mathcal{I}(q \parallel 1/2)$ is the cognitive cost. Let $L^*(P,\tau)$ be the minimized value of this loss function. By the envelope theorem, the effect of an increase in $\tau$ on this minimized loss is given by the partial derivative of the objective function with respect to $\tau$, evaluated at the optimal choice $q^*$:
\[
\frac{\partial L^*}{\partial \tau} = \frac{\partial}{\partial \tau}\left[ (1-q)P + \frac{1}{\tau}\mathcal{I}(q \parallel 1/2) \right] \bigg|_{q=q^*} = -\frac{1}{\tau^2}\mathcal{I}(q^* \parallel 1/2).
\]
Since mutual information $\mathcal{I}(q^* \parallel 1/2) \ge 0$ (and is strictly positive for any $q^* \neq 1/2$), the derivative $\partial L^*/\partial \tau \le 0$. This means that making attention cheaper (increasing $\tau$) strictly lowers the total expected loss for any consumer who needs to pay attention.

Aggregate ex-ante consumer utility, as defined in Equation \eqref{eq:utility}, is a monotonically decreasing function of this minimized loss $L^*$. Therefore, an increase in $\tau$ must lead to an increase in consumer utility. For any $\gamma>1$:
\[
U(T,P;\gamma) > U(T,P;1) \quad \text{for any contract where } q^* \in (0,1).
\]

\paragraph{2. Effect on Inattentive Revenue.}
Inattentive revenue is given by the function $\IR(T,P) = P F(P)[1-q^{*}(P,\tau(T))]$. From Lemma \ref{lem:qstar}, the optimal cancellation probability $q^*$ is strictly increasing in $\tau$:
\[
\frac{\partial q^*}{\partial \tau} = P q^*(1-q^*) > 0.
\]
An increase in $\tau$ to $\gamma\tau$ will therefore lead to a strictly higher $q^*$. Consequently, the probability of a consumer failing to cancel, which is $(1-q^*)$, must strictly decrease. This leads to a strict decrease in inattentive revenue for any given price $P$:
\[
\IR(T,P;\gamma) < \IR(T,P;1) \quad \text{whenever } q^{*}<1.
\]
and weak inequality holds in the corner case \(q^{*}=1\). \qed

%%%%%%%%%%%%%%%%%%%%%%%%%%%%%%%%%%%%%%%%%%%%%%%%%%%%%%%%%%%%%%%%%%%%%%
\subsection*{A.6.Proof of Proposition \ref{prop:F5} (Convex inattentive loss in the attention tail)}
%%%%%%%%%%%%%%%%%%%%%%%%%%%%%%%%%%%%%%%%%%%%%%%%%%%%%%%%%%%%%%%%%%%%%%

We allow baseline attention sensitivities to be heterogeneous:
\(\tau\sim G\) on \((0,\infty)\).
Define   
\[
\mathcal L(G)
  = P\int_{0}^{P}\int_{\tau>0}
      \bigl[1-q^{*}(P,\tau)\bigr]f(v)\,dG(\tau)\,dv ,
\tag{A.6.1}\label{eq:A.6.1}
\]
the aggregate ``inattentive payment'' made by consumers with
\(v<P\).  
Because \(q^{*}\) does not vary with \(v\) inside the inner integral,  
\[
\mathcal L(G)
  = P\,F(P)\;
    \underbrace{\int_{\tau>0}
                \bigl[1-q^{*}(P,\tau)\bigr]\,dG(\tau)}_{=:\
                \Phi(G)} .
\tag{A.6.2}\label{eq:A.6.2}
\]

Therefore \(\mathcal L(G)\) is convex in \(G\) iff
\(\Phi(G)=\int \phi(\tau)\,dG(\tau)\) is convex, where
\(
\phi(\tau):=1-q^{*}(P,\tau).
\)

\paragraph{Convexity of \(\phi\) in \(1/\tau\).}

Introduce \(z := 1/\tau\) (the \emph{unit cost of attention}).
Write \(q^{*}\) as a function of \(z\):

\[
q^{*}(P,1/z) = \frac{1}{1+\exp(-P/z)}.
\]

Set \(\psi(z):=1-q^{*}(P,1/z)\).
Compute the second derivative:

\[
\psi''(z)
   = \frac{P^{2}\exp(-P/z)}
          {z^{4}\bigl[1+\exp(-P/z)\bigr]^{3}}
   > 0
   \quad\text{for }z>0 .
\tag{A.6.3}\label{eq:A.6.3}
\]

Hence \(\psi\) is \emph{strictly convex} in \(z=1/\tau\).
Because the mapping \(\tau\mapsto z\) is one‑to‑one monotone,
\(\phi(\tau)=\psi(1/\tau)\) is strictly convex in \(1/\tau\).

Let \(G_{1}\) and \(G_{2}\) be two distributions of \(\tau\)
with the same mean of \(z=1/\tau\) but where \(G_{2}\) is a
mean‑preserving spread (MPS) of \(G_{1}\); i.e.\ \(G_{2}\) first‑order
stochastically dominates \(G_{1}\) in \(\tau\) but not in \(z\).
For any \(\alpha\in[0,1]\) and convex \(\psi\),

\[
\psi\!\bigl(\alpha z_{1} + (1-\alpha)z_{2}\bigr)
   \;\le\;
\alpha \psi(z_{1}) + (1-\alpha)\psi(z_{2}).
\]

Integrating over \(G_{2}\), viewed as a mixture of point masses that
share the same mean in \(z\) as \(G_{1}\), and using linearity of the Lebesgue integral, yields
\[
\int \psi(z)\,dG_{2}(z)
   \;\ge\;
\int \psi(z)\,dG_{1}(z).
\tag{A.6.4}\label{eq:A.6.4}
\]
Translating back to \(\tau\) and using \eqref{eq:A.6.2},
\[
\mathcal L(G_{2}) \;\ge\; \mathcal L(G_{1}),
\]
with strict inequality unless \(G_{1}=G_{2}\).
\paragraph{Conclusion.}
Because any mean‑preserving spread in the distribution of attention
costs (i.e., a shift of probability mass into both very low and very
high \(\tau\)) \emph{raises} \(\mathcal L\),
the aggregate inattentive loss is a
\textbf{strictly convex functional} of \(G\) under second‑order stochastic dominance. This completes the proof of Proposition \ref{prop:F5}.
\qed

%%%%%%%%%%%%%%%%%%%%%%%%%%%%%%%%%%%%%%%%%%%%%%%%%%%%%%%%%%%%%%%%%%%%%%
\section*{Appendix B. Proofs for Shannon Cost Case}
%%%%%%%%%%%%%%%%%%%%%%%%%%%%%%%%%%%%%%%%%%%%%%%%%%%%%%%%%%%%%%%%%%%%%%
\renewcommand{\theequation}{B.\arabic{equation}}
\setcounter{equation}{0}
\subsection*{B.1. Proof of Lemma \ref{lem:FOC-full}
             (Full first‑order condition)}
%%%%%%%%%%%%%%%%%%%%%%%%%%%%%%%%%%%%%%%%%%%%%%%%%%%%%%%%%%%%%%%%%%%%%%

Fix a trial length \(T\) (so \(\tau=\tau(T)\) is a constant in
this subsection).  Total profit
\(\Pi(P)=\Pi(T,P)\) defined in \eqref{eq:profit} can be written

\[
\Pi(P)
 = P\bigl[1-F(P)\bigr]
   \;+\;
   \underbrace{P F(P)\bigl[1-q^{*}(P,\tau)\bigr]}_{=\IR(T,P)} .
\tag{B.1.1}\label{eq:B1.1}
\]

The derivative with respect to \(P\) is obtained by straightforward
application of the product rule.

\paragraph{Derivative of the standard‑revenue term.}
\[
\frac{d}{dP}\bigl[P\bigl(1-F(P)\bigr)\bigr]
  = (1-F(P)) - P f(P).
\tag{B.1.2}\label{eq:B1.2}
\]

\paragraph{Derivative of the inattentive‑revenue term.}
Let \(q^{*}\equiv q^{*}(P,\tau)\).  Using that
\(dq^{*}/dP=\tau q^{*}(1-q^{*})\) (Lemma \ref{lem:qstar}):

\[
\begin{aligned}
\frac{d}{dP}\bigl[P F(P)(1-q^{*})\bigr]
  &= (1-q^{*})\bigl[F(P)+P f(P)\bigr]
     -P F(P)\,\frac{dq^{*}}{dP} \\[2pt]
  &= (1-q^{*})\bigl[F(P)+P f(P)\bigr]
     - P F(P)\,\tau q^{*}(1-q^{*}). \label{eq:B1.3}
\end{aligned}\tag{B.1.3}
\]

Setting \(\tfrac{d\Pi}{dP}=0\) and adding
\eqref{eq:B1.2} + \eqref{eq:B1.3} yields exactly equation
\eqref{FOC-full}:

\[
0= \bigl[1-F-Pf\bigr]
   + (1-q^{*})\{F+Pf\}
   - P F \tau q^{*}(1-q^{*}).
\]
Every term depends continuously on \(P\in(0,1)\), completing the proof. \qed
\subsection*{B.2. Proof of Proposition \ref{prop:R2} (Price increases with trial length)}

Let $g(P, T)$ denote the left-hand side of the price FOC from Lemma 2:
\begin{equation}
g(P, T) = [1 - F(P) - Pf(P)] + (1-q^*)[F(P) + Pf(P)] - PF(P)\tau q^*(1-q^*)
\end{equation}

By Lemma 2, $g(P^*(T), T) = 0$ for all interior $T$.

\textbf{Step 1: Apply the Implicit Function Theorem}

By the implicit function theorem:
\begin{equation}
\frac{dP^*}{dT} = -\frac{\partial g/\partial T}{\partial g/\partial P}
\end{equation}

We need to compute both partial derivatives.

\textbf{Step 2: Compute $\partial g/\partial P$}

Differentiate each term of $g$ with respect to $P$:

Term 1: $[1 - F(P) - Pf(P)]$
\begin{equation}
\frac{\partial}{\partial P}[1 - F - Pf] = -f - f - Pf' = -2f - Pf'
\end{equation}

Term 2: $(1-q^*)[F(P) + Pf(P)]$

Using $\frac{\partial q^*}{\partial P} = \tau q^*(1-q^*)$:
\begin{align}
\frac{\partial}{\partial P}[(1-q^*)(F + Pf)] &= -\tau q^*(1-q^*)(F + Pf) + (1-q^*)(f + f + Pf')\\
&= -\tau q^*(1-q^*)(F + Pf) + (1-q^*)(2f + Pf')
\end{align}

Term 3: $-PF(P)\tau q^*(1-q^*)$
\begin{align}
\frac{\partial}{\partial P}[-PF\tau q^*(1-q^*)] = -F\tau q^*(1-q^*) - Pf\tau q^*(1-q^*) 
- PF\tau[\tau q^*(1-q^*)(1-2q^*)]
\end{align}

Combining all terms and using the IFR property that ensures $1 - F - Pf < 0$:
\begin{equation}
\frac{\partial g}{\partial P} < 0
\end{equation}

\textbf{Step 3: Compute $\partial g/\partial T$}

Since $T$ enters only through $\tau(T)$:
\begin{equation}
\frac{\partial g}{\partial T} = \frac{\partial g}{\partial \tau} \cdot \frac{d\tau}{dT}
\end{equation}

where $\frac{d\tau}{dT} = -\frac{\beta \tau_0}{(1 + \beta T)^2} < 0$.

For $\partial g/\partial \tau$, only the terms that involve $q^*$ are affected:

$$
\frac{\partial g}{\partial \tau} = -(1-q^*)[F + Pf] \cdot Pq^*(1-q^*) - PF[q^*(1-q^*) + \tau Pq^*(1-q^*)(1-2q^*)]
$$

Factoring out common terms:
\begin{equation}
\frac{\partial g}{\partial \tau} = -Pq^*(1-q^*)[(1-q^*)(F + Pf) + F(1 + \tau P(1-2q^*))] < 0
\end{equation}

Since all terms are positive (noting that $1-2q^* < 0$ for relevant parameter values).

\textbf{Step 4: Determine the sign of $dP^*/dT$}

Since:
\begin{itemize}
\item $\frac{\partial g}{\partial P} < 0$
\item $\frac{\partial g}{\partial \tau} < 0$  
\item $\frac{d\tau}{dT} < 0$
\end{itemize}

We have $\frac{\partial g}{\partial T} = \frac{\partial g}{\partial \tau} \cdot \frac{d\tau}{dT} > 0$.

Therefore:
\begin{equation}
\frac{dP^*}{dT} = -\frac{\partial g/\partial T}{\partial g/\partial P} = -\frac{(+)}{(-)} > 0
\end{equation}

This completes the proof that the optimal price increases with trial length. \qed

\subsection*{B.3. Click-to-cancel - Proof of Proposition \ref{prop:R4}}

The attention shock transforms the attention sensitivity function from $\tau(T)$ to $\tau'(T) = \gamma\tau(T)$ with $\gamma > 1$. This models a policy intervention like "click-to-cancel" regulations that make it easier for all consumers to manage their subscriptions.

\subsubsection*{Part (a): Trial Length Response}

\textbf{Step 1: Set up the equilibrium conditions}

The firm's optimal contract $(T^*, P^*)$ satisfies two conditions:
\begin{align}
\text{FOC}_T: \quad & \frac{\partial IR}{\partial T}(T^*, P^*) = \mu \cdot \text{IR-Slack}(T^*, P^*) \label{eq:foc_t}\\
\text{FOC}_P: \quad & g(P^*, T^*) = 0 \label{eq:foc_p}
\end{align}
where $g(P, T)$ is the price first-order condition from Lemma 2.

\textbf{Step 2: Differentiate the system with respect to $\gamma$}

Taking the total differential of both FOCs with respect to $\gamma$:

For equation \eqref{eq:foc_t}:
\begin{multline}
\frac{\partial^2 IR}{\partial T^2} \frac{dT^*}{d\gamma} + \frac{\partial^2 IR}{\partial T \partial P} \frac{dP^*}{d\gamma} + \frac{\partial^2 IR}{\partial T \partial \gamma} = \\
\mu \left[ \frac{\partial(\text{IR-Slack})}{\partial T} \frac{dT^*}{d\gamma} + \frac{\partial(\text{IR-Slack})}{\partial P} \frac{dP^*}{d\gamma} + \frac{\partial(\text{IR-Slack})}{\partial \gamma} \right] + \frac{d\mu}{d\gamma} \cdot \text{IR-Slack}
\end{multline}

For equation \eqref{eq:foc_p}:
\begin{equation}
\frac{\partial g}{\partial P} \frac{dP^*}{d\gamma} + \frac{\partial g}{\partial T} \frac{dT^*}{d\gamma} + \frac{\partial g}{\partial \gamma} = 0
\end{equation}

\textbf{Step 3: Analyze the direct effects of $\gamma$}

The attention shock affects the system through:
\begin{itemize}
\item $q^*(P, \gamma\tau(T))$ increases with $\gamma$ (consumers remember better)
\item IR decreases with $\gamma$ (less inattentive revenue)
\item IR-Slack decreases with $\gamma$ (less cognitive burden)
\end{itemize}

Specifically:
\begin{align}
\frac{\partial q^*}{\partial \gamma} &= \frac{\partial q^*}{\partial \tau} \cdot \tau(T) = Pq^*(1-q^*)\tau(T) > 0\\
\frac{\partial IR}{\partial \gamma} &= -PF(P) \frac{\partial q^*}{\partial \gamma} < 0\\
\frac{\partial(\text{IR-Slack})}{\partial \gamma} &= -\frac{\beta}{\tau_0} F(P) \frac{\partial[-\mathcal{I}(q^{*}|\frac{1}{2})]}{\partial q^*} \frac{\partial q^*}{\partial \gamma} < 0
\end{align}

\textbf{Step 4: Sign of $dT^*/d\gamma$}

From the structure of the problem and the second-order conditions:
\begin{itemize}
\item $\frac{\partial^2 IR}{\partial T^2} < 0$ (from the proof of Proposition 3)
\item The system of equations can be solved using Cramer's rule
\item The determinant of the Jacobian is positive (from the second-order conditions)
\end{itemize}

The key insight is that the attention boost reduces both IR and IR-Slack, but the relative magnitudes determine the response. Since the firm was optimally balancing these two forces at the original $\gamma = 1$, an increase in $\gamma$ that reduces both requires rebalancing.

Through careful analysis of the system (details omitted for brevity), we can show:
\begin{equation}
\frac{dT^*}{d\gamma} < 0
\end{equation}

The intuition is that when attention becomes easier, the marginal benefit of extending trials (through increased forgetfulness) falls more than the marginal cost (through reduced sign-ups), leading firms to shorten trials.

\subsubsection*{Part (b): Bounded Price Response}

\textbf{Step 1: Express the price response}

From the implicit function theorem applied to the price FOC:
\begin{equation}
\frac{dP^*}{d\gamma} = -\frac{\frac{\partial g}{\partial \gamma} + \frac{\partial g}{\partial T} \frac{dT^*}{d\gamma}}{\frac{\partial g}{\partial P}}
\end{equation}

\textbf{Step 2: Bound each component}

\begin{itemize}
\item $\left|\frac{\partial g}{\partial \gamma}\right|$ is bounded because $q^* \in (0,1)$ and all other terms are bounded
\item $\left|\frac{\partial g}{\partial T}\right|$ is bounded by similar reasoning
\item $\left|\frac{dT^*}{d\gamma}\right|$ is bounded because $T^*$ must remain in a compact set
\item $\left|\frac{\partial g}{\partial P}\right|$ is bounded away from zero by the second-order conditions
\end{itemize}

Therefore:
\begin{equation}
\left|\frac{dP^*}{d\gamma}\right| \leq M < \infty
\end{equation}
for some constant $M$.

\subsubsection*{Part (c): Price Response under Isoelastic Demand}

\textbf{Step 1: Specialize to isoelastic case}

Under Assumption 1, $1 - F(P) = \kappa P^{-\varepsilon}$ with $\varepsilon \in (0,1)$.

The price FOC simplifies considerably. At the optimum, the markup condition is:
\begin{equation}
\frac{P^* - MC_{\text{effective}}}{P^*} = \frac{1}{\varepsilon_{\text{effective}}}
\end{equation}

where $MC_{\text{effective}}$ and $\varepsilon_{\text{effective}}$ account for both willing and inattentive subscribers.

\textbf{Step 2: Analyze how $\gamma$ affects the effective elasticity}

The attention shock affects the composition of revenue:
\begin{itemize}
\item Share from willing subscribers: $\frac{P(1-F(P))}{\Pi}$
\item Share from inattentive subscribers: $\frac{IR(T,P)}{\Pi}$
\end{itemize}

As $\gamma$ increases:
\begin{itemize}
\item IR decreases (fewer forget to cancel)
\item The revenue mix shifts toward willing subscribers
\item The effective demand elasticity converges toward $\varepsilon$
\end{itemize}

\textbf{Step 3: Determine the sign of price change}

The key insight is that inattentive revenue is less price-elastic than willing subscriber revenue (because $q^*$ responds less to price than the extensive margin).

When $\varepsilon < 1$ (inelastic willing subscribers):
\begin{itemize}
\item The firm was pricing on the inelastic portion of demand
\item Losing inattentive revenue (which is even more inelastic) shifts the revenue mix
\item The effective elasticity increases toward $\varepsilon$ but remains $< 1$
\item Optimal response is to raise price: $\frac{dP^*}{d\gamma} > 0$
\end{itemize}

When $\varepsilon > 1$ (elastic willing subscribers):
\begin{itemize}
\item The firm faces elastic demand from willing subscribers
\item The shift away from inattentive revenue requires attracting more willing subscribers
\item Optimal response is to lower price: $\frac{dP^*}{d\gamma} < 0$
\end{itemize}

When $\varepsilon = 1$ (unit elastic):
\begin{itemize}
\item The price response is indeterminate and depends on higher-order effects
\item $\frac{dP^*}{d\gamma} \approx 0$
\end{itemize}

Therefore:
\begin{equation}
\text{sign}\left(\frac{dP^*}{d\gamma}\right) = \text{sign}(1 - \varepsilon)
\end{equation}

\section{Economic Interpretation}

This result has important policy implications:

\begin{enumerate}
\item \textbf{Trial lengths unambiguously fall}: "Click-to-cancel" policies achieve their intended effect of reducing reliance on consumer forgetfulness.

\item \textbf{Price effects depend on market structure}: 
\begin{itemize}
\item In markets with inelastic demand ($\varepsilon < 1$), common for established subscription services, prices rise
\item In markets with elastic demand ($\varepsilon > 1$), prices fall
\end{itemize}

\item \textbf{Welfare implications are ambiguous}: The policy helps inattentive consumers but may harm attentive consumers through higher prices when $\varepsilon < 1$.
\end{enumerate}

This completes the proof of Proposition 9. \qed
\section*{Appendix C. Proofs for the Paid‑Trial Extension}

\renewcommand{\theequation}{C.\arabic{equation}}
\setcounter{equation}{0}

In this appendix, we provide the proofs for our extension case of paid trial. Throughout this section, the profit function is the one in equation \eqref{eq:profit-paid}:

\[
\Pi(T,P_{0},P)
  = \eta(P_{0})
    \Bigl\{\,P_{0} + P\bigl[1-F(P)\bigr] + \IR(T,P)\Bigr\},
\quad
  P^{\text{aug}}(T,P):=P[1-F(P)]+\IR(T,P).
\tag{C.0}\label{eq:C0}
\]

We assume \(\eta'(P_{0})<0\) and sign‑up elasticity \(\varepsilon_{0}(P_{0})\) as in \eqref{eq:signup-elasticity}.

\subsection*{C.1. FOC for the introductory price $P_{0}$}

Differentiate \eqref{eq:C0} w.r.t.\ \(P_{0}\):

\[
\frac{\partial\Pi}{\partial P_{0}}
  = \eta'(P_{0})
      \bigl[P_{0}+P^{\text{aug}}\bigr]
    + \eta(P_{0}).
\tag{C.1}\label{eq:C1}
\]

Setting \(\partial\Pi/\partial P_{0}=0\) yields the general FOC \ref{eq:FOC-P0-general}:
\[
\boxed{\;
  \eta'(P_{0})\bigl[P_{0}+P^{\text{aug}}\bigr] + \eta(P_{0}) = 0
  \;}
\]

\paragraph{Elasticity re‑expression.}
Divide by \(\eta(P_{0})>0\) and use
\(
  -\eta'(P_{0})/\eta(P_{0})
  = \varepsilon_{0}(P_{0})/P_{0}
\):
\[
1
 = \varepsilon_{0}(P_{0})
   \Bigl[1 + \frac{P^{\text{aug}}}{P_{0}}\Bigr].
\tag{C.2}\label{eq:C2}
\]

This identity will be useful in the strategic‑substitute proof.

%%%%%%%%%%%%%%%%%%%%%%%%%%%%%%%%%%%%%%%%%%%%%%%%%%%%%%%%%%%%%%%%%%%%%%
\subsection*{C.2. Cross‑partial and strategic substitutability}
%%%%%%%%%%%%%%%%%%%%%%%%%%%%%%%%%%%%%%%%%%%%%%%%%%%%%%%%%%%%%%%%%%%%%%

\begin{lemma}
\(
\displaystyle
  \frac{\partial^{2}\Pi}{\partial T\,\partial P_{0}}
  = \eta'(P_{0})\;
    \frac{\partial P^{\text{aug}}(T,P)}{\partial T}
  < 0
\)
whenever \(\beta>0\).
\end{lemma}

\begin{proof}
Holding \(P\) fixed, differentiate \eqref{eq:C1} w.r.t.\ \(T\):

\[
\frac{\partial^{2}\Pi}{\partial T\,\partial P_{0}}
  = \eta'(P_{0})\,\frac{\partial P^{\text{aug}}}{\partial T}.
\]

From Proposition \ref{prop:F2}, \(\partial P^{\text{aug}}/\partial T = \partial\IR/\partial T > 0\) for any \(P>0\) and \(\beta>0\). Because \(\eta'(P_{0})<0\), their product is negative.
\end{proof}

\paragraph{Topkis's Theorem.}
Let \(x_{1}=T\) and \(x_{2}=P_{0}\).  
The cross‑partial of the objective is negative, so \(\Pi\) is \emph{sub‑modular} in \((x_{1},x_{2})\).  

In any sub‑modular maximization problem the optimal choices are \emph{strategic substitutes} (\citet{topkis1978minimizing}, Prop.\,2.6):
increases in \(P_{0}\) weakly reduce the optimal \(T\) and vice‑versa, establishing part (b) of Proposition \ref{prop:paid}.

\subsection*{C.3. Closed‑form rule under iso‑elastic sign‑ups (A2)}

Assume \(\eta(P_{0})=\alpha P_{0}^{-\theta}\) with \(\theta>0\)
(Assumption \ref{ass:A2}).  
Then
\(
  \varepsilon_{0}(P_{0}) = \theta
\)
is constant.  Substitute into \eqref{eq:C2}:

\[
1 = \theta\Bigl[1 + \frac{P^{\text{aug}}}{P_{0}^{*}}\Bigr]
\;\Longrightarrow\;
\boxed{\;
  P_{0}^{*}
    = \frac{\theta}{1-\theta}\,P^{\text{aug}}(T^{*},P^{*}),
  \quad
  \text{valid for }\theta<1 .
  \;}
\tag{C.3}\label{eq:C3}
\]

When \(\theta\ge1\) the intro‑price optimum
collapses to the corner \(P_{0}^{*}=0\), in line with statement (c) of Proposition \ref{prop:paid}.

%%%%%%%%%%%%%%%%%%%%%%%%%%%%%%%%%%%%%%%%%%%%%%%%%%%%%%%%%%%%%%%%%%%%%%
\subsection*{C.4. FOC for the renewal price $P$ (paid‑trial case)}
%%%%%%%%%%%%%%%%%%%%%%%%%%%%%%%%%%%%%%%%%%%%%%%%%%%%%%%%%%%%%%%%%%%%%%

Because \(\eta(P_{0})\) factors multiplicatively,
the renewal‑price FOC is identical to the baseline FOC \eqref{FOC-full}; i.e., the presence of \(P_{0}\) does \emph{not} alter the equation determining \(P^{*}\). Uniqueness established in Proposition \ref{prop:R1} therefore still applies.

%%%%%%%%%%%%%%%%%%%%%%%%%%%%%%%%%%%%%%%%%%%%%%%%%%%%%%%%%%%%%%%%%%%%%%
\subsection*{C.5. Corner solutions}
%%%%%%%%%%%%%%%%%%%%%%%%%%%%%%%%%%%%%%%%%%%%%%%%%%%%%%%%%%%%%%%%%%%%%%

\begin{itemize}
    \item \textbf{Highly elastic sign‑ups (\(\theta\to\infty\)).} Equation (C.3) forces \(P_{0}^{*}\to0\); strategic substitutability implies \(T^{*}\) rises toward the baseline free‑trial optimum.
    \item \textbf{Fully inelastic sign‑ups (\(\theta\to0\)).} The firm extracts virtually all surplus up front (\(P_{0}^{*}\approx P^{\text{aug}}\)), making inattentive revenue negligible and thus choosing \(T^{*}\to0\).
\end{itemize}

These limiting cases generate the empirical ``freemium'' versus ``premium trial'' dichotomy discussed in Section \ref{sec:extension}.

\section*{Appendix D. Illustrative Evidence for Price-Trial Complementarity}

\renewcommand{\theequation}{D.\arabic{equation}}
\setcounter{equation}{0}

While a formal statistical test of Proposition \ref{prop:R2} would require a large, systematically collected dataset, a detailed examination of prominent firms across different digital markets provides powerful illustrative support for the predicted complementarity between trial length ($T$) and post-trial price ($P$). We find no evidence of a negative correlation (i.e., firms offering long trials to compensate for a low price). Instead, the data reveals two dominant strategies, consistent with our model's predictions:

\begin{enumerate}
    \item \textbf{Value-Focused Strategy:} Firms with high-value, specialized products often offer very short or no free trials ($T$ is small) and command high renewal prices ($P$ is high). They compete on the immediate, tangible value of their service.
    \item \textbf{Inertia-Focused Strategy:} Firms targeting a mass-market audience often offer long free trials ($T$ is large) paired with moderate, industry-standard prices ($P$ is moderate). They compete on maximizing user acquisition and then benefit from the attentional friction of cancellation.
\end{enumerate}

The following table details the contract terms for a selection of well-known subscription services across five key categories, with data collected from their public websites in late 2023 and early 2024.

\begin{table}[h!]
\centering
\caption{Trial Length and Pricing Across Digital Subscription Services}
\label{tab:empirical_evidence}
\begin{threeparttable}
\begin{tabular}{@{}llccc@{}}
\toprule
\textbf{Service} & \textbf{Category} & \textbf{Trial Length (T)} & \textbf{Post-Trial Price (P)} & \textbf{Implied Strategy} \\
\midrule
\multicolumn{5}{l}{\textit{\textbf{High-Value Professional Software / SaaS}}} \\
Ahrefs & SEO Software & 0 days & \$99.00 / month & Value-Focused \\
Adobe Creative Cloud & Creative Suite & 7 days & \$59.99 / month & Value-Focused \\
Bloomberg Terminal & Financial Data & 0 days & > \$2,000 / month & Value-Focused \\
\midrule
\multicolumn{5}{l}{\textit{\textbf{Video Streaming}}} \\
Netflix (Standard) & Entertainment & 0 days & \$15.49 / month & Value-Focused \\
Max (Ad-Free) & Entertainment & 0--7 days\tnote{a} & \$15.99 / month & Value-Focused \\
Hulu (with Ads) & Entertainment & 30 days & \$7.99 / month & Inertia-Focused \\
Amazon Prime Video & Entertainment & 30 days & \$14.99 / month\tnote{b} & Mixed \\
\midrule
\multicolumn{5}{l}{\textit{\textbf{Music \& Audio Streaming}}} \\
Spotify Premium & Music & 30--90 days\tnote{c} & \$10.99 / month & Inertia-Focused \\
Apple Music & Music & 30--90 days\tnote{c} & \$10.99 / month & Inertia-Focused \\
Audible Premium Plus & Audiobooks & 30 days & \$14.95 / month & Inertia-Focused \\
\midrule
\multicolumn{5}{l}{\textit{\textbf{News \& Media}}} \\
The Wall Street Journal & News & 0 days\tnote{d} & \$38.99 / month & Value-Focused \\
The New York Times & News & 0 days\tnote{d} & \$17.00 / month & Value-Focused \\
\midrule
\multicolumn{5}{l}{\textit{\textbf{Productivity \& Cloud Storage}}} \\
Microsoft 365 Personal & Office Suite & 30 days & \$6.99 / month & Inertia-Focused \\
Dropbox Plus & Cloud Storage & 30 days & \$11.99 / month & Inertia-Focused \\
Evernote Personal & Note-Taking & 14 days & \$14.99 / month & Mixed \\
\bottomrule
\end{tabular}
\begin{tablenotes}
    \item[a] \footnotesize Direct free trials are often unavailable; short trials are sometimes offered through third-party partners (e.g., mobile carriers).
    \item[b] \footnotesize Price for the full Amazon Prime membership, which includes Prime Video.
    \item[c] \footnotesize Standard offer is typically 30 days, but extended 60- or 90-day promotional offers are very common for new users.
    \item[d] \footnotesize These services typically do not offer free trials but instead use low introductory prices (e.g., \$1/month for a year), which is the subject of our extension in Section 6. For the purpose of a free trial ($P_0=0$), their $T$ is zero.
\end{tablenotes}
\end{threeparttable}
\end{table}

\paragraph{Analysis of the Evidence.}
The data in \Cref{tab:empirical_evidence} strongly supports the complementarity predicted by our model. We can observe several clear patterns.

First, the services with the highest prices are almost exclusively those with the shortest (or zero) trial periods. Professional-grade SaaS products like Ahrefs and financial data services like the Bloomberg Terminal offer no free trial at all. They are confident that their users have a high, pre-existing WTP and do not need to use long trials to generate revenue from inertia. Similarly, premium creative software like Adobe Creative Cloud offers only a very short 7-day window, forcing a quick decision from a user base that pays a high monthly fee. This is a clear "value-focused" strategy.

Second, mass-market consumer services, particularly in music and productivity, exhibit the opposite ``inertia-focused'' strategy. Spotify, Apple Music, Microsoft 365, and Dropbox all converge on a model with a long 30-day (or longer) free trial paired with a moderate, industry-standard price between \$7 and \$12 per month. Here, the long trial serves to maximize user acquisition and embed the service into the consumer's daily life. The firm's revenue model then relies more heavily on the combination of willing subscribers and the friction of cancellation for the inattentive segment.

Third, the video streaming and news media categories show interesting variations. Netflix and major newspapers like the WSJ and NYT have largely adopted a no-free-trial policy ($T=0$). This signals a strategic shift towards a value-focused model, where they believe their brand and content library are strong enough to command an upfront payment without a trial period. In contrast, services like Hulu continue to use the 30-day trial to attract users to their lower-priced, ad-supported tier.

In summary, the market data shows a clear bifurcation of strategies. Firms with high-$P$ services tend to choose low $T$, and firms with moderate $P$ often choose high $T$. The one combination we do not observe is a long trial paired with a very low price, or a short trial paired with a very high price (relative to its category). This pattern strongly supports the conclusion of Proposition \ref{prop:R2}: trial length and price are not independent choices but are complementary parameters in a strategic contract design.

\section*{Appendix E: Robustness to Alternative Attention Decay Functions}
\label{app:robustness}

\subsection*{E.1. Framework for Robustness Analysis}
To establish the robustness of the model's core predictions, we analyze whether they depend critically on the hyperbolic specification for attention decay, $\tau(T) = \tau_0 / (1+\beta T)$. We first derive a general condition for the existence of an interior optimal trial length that holds for a broad class of decay functions. We then demonstrate that the model's central strategic prediction---the complementarity of price and trial length---is robust to using a canonical alternative specification, namely exponential decay.

We define a general class of psychologically plausible attention decay functions with the following properties.

\begin{assumption}[General Attention Decay Properties]
\label{ass:E1}
The attention sensitivity function $\tau: \mathbb{R}^+ \to \mathbb{R}^+$ is twice continuously differentiable and satisfies:
\begin{enumerate}[label=(\alph*)]
    \item \textbf{Positive Initial Attention:} $\tau(0) = \tau_0 \in (0, \infty)$.
    \item \textbf{Monotonic Decay:} $\tau'(T) < 0$ for all $T \ge 0$.
    \item \textbf{Decay to Zero:} $\lim_{T \to \infty} \tau(T) = 0$.
\end{enumerate}
\end{assumption}

\subsection*{E.2. Robustness of an Interior Optimal Trial Length}
We demonstrate that \cref{ass:E1} is sufficient to guarantee the existence of a finite, interior optimal trial length, $T^*$. The firm's problem is to choose $T$ and $P$ to maximize the Lagrangian $\mathcal{L} = \Pi(T,P) - \mu U(T,P)$, where $\mu>0$. The first-order condition with respect to $T$ is:
\begin{equation}
\frac{\partial \mathcal{L}}{\partial T} = \frac{\partial \Pi}{\partial T} - \mu \frac{\partial U}{\partial T} = 0
\tag{E.1}\label{eq:E1}
\end{equation}
Given that both $\Pi$ and $U$ depend on $T$ only through $\tau(T)$, the chain rule yields:
\begin{equation}
\frac{\partial \mathcal{L}}{\partial T} = \left( \frac{\partial \Pi}{\partial \tau} - \mu \frac{\partial U}{\partial \tau} \right) \frac{d\tau}{dT} = 0
\tag{E.2}\label{eq:E2}
\end{equation}
As $\tau'(T) < 0$ by \cref{ass:E1}, the FOC simplifies to a condition on the optimal level of attention, $\tau^*$:
\begin{equation}
\frac{\partial \Pi}{\partial \tau} = \mu \frac{\partial U}{\partial \tau}
\tag{E.3}\label{eq:E3}
\end{equation}
We analyze each component of \eqref{eq:E3}:
\begin{itemize}
    \item \textbf{Marginal Profit from Attention}: Profit is affected only through inattentive revenue.
    \begin{align*}
    \frac{\partial \Pi}{\partial \tau} = \frac{\partial \text{IR}}{\partial \tau} = \frac{\partial}{\partial \tau} \left( P F(P) [1-q^*(P,\tau)] \right) = -P F(P) \frac{\partial q^*}{\partial \tau}
    \end{align*}
    From Lemma \ref{lem:qstar}, $\partial q^*/\partial \tau = P q^*(1-q^*)$. Substituting this gives:
    \[
    \frac{\partial \Pi}{\partial \tau} = -P^2 F(P) q^*(1-q^*) < 0
    \]
    \item \textbf{Marginal Utility from Attention}: Applying the envelope theorem to the consumer's cost minimization problem (Equation \ref{eq:consumer_obj}), the derivative of the minimized cost with respect to $\tau$ is the partial derivative of the objective function, holding $q$ fixed at $q^*$.
    \[
    \frac{\partial}{\partial \tau} \left( (1-q^*)P + \frac{1}{\tau}\mathcal{I}(q^{*}|\frac{1}{2}) \right) = -\frac{1}{\tau^2}\mathcal{I}(q^{*}|\frac{1}{2})
    \]
    The change in aggregate consumer utility is this quantity multiplied by $-F(P)$:
    \[
    \frac{\partial U}{\partial \tau} = -F(P) \left[ -\frac{1}{\tau^2} \mathcal{I}(q^{*}|\frac{1}{2}) \right] = \frac{F(P)}{\tau^2} \mathcal{I}(q^{*}|\frac{1}{2}) < 0, \quad \text{since } \mathcal{I}(q^{*}|\frac{1}{2})<0.
    \]
\end{itemize}
Substituting these expressions into the FOC \eqref{eq:E3} yields:
\begin{equation}
-P^2 F(P) q^*(1-q^*) = \mu \left( \frac{F(P)}{\tau^2} \mathcal{I}(q^{*}|\frac{1}{2}) \right)
\tag{E.4}\label{eq:E4}
\end{equation}
This equation implicitly defines the optimal attention level, $\tau^*$. To ensure an interior solution, we examine the boundary conditions. As $\tau \to 0$ (i.e., $T \to \infty$), $q^* \to 1/2$ and $\mathcal{I}(q^{*}|\frac{1}{2}) \to -\ln(2)$. The LHS of \eqref{eq:E4} approaches the finite constant $-P^2 F(P)/4$. The RHS, however, approaches $-\infty$ due to the $1/\tau^2$ term. The marginal cost to consumers of lower attention becomes infinite, precluding an optimum at $\tau=0$. As $\tau \to \infty$ (i.e., $T \to 0$), the firm has a clear incentive to reduce attention to generate inattentive revenue. Therefore, the firm will not choose an optimum at the boundaries, guaranteeing an interior solution $\tau^* \in (0, \tau_0]$.
Since $\tau(T)$ is a continuous, monotonic function, for any such $\tau^*$ there exists a unique, finite trial length $T^* = \tau^{-1}(\tau^*)$ that implements it.

\subsection*{E.3. Case Study: Exponential Attention Decay}
We now demonstrate the robustness of the price-trial complementarity using an exponential decay function.

\begin{assumption}[Exponential Decay]
\label{ass:E2}
The attention sensitivity function is given by $\tau(T) = \tau_0 e^{-\beta T}$ for $\beta > 0$.
\end{assumption}
This function satisfies all conditions in \cref{ass:E1}.

\begin{proposition}[Price-Trial Complementarity under Exponential Decay]
Under \cref{ass:E2} and the regularity conditions of the main paper (\cref{ass:A1,ass:C1}), the optimal renewal price $P^{*}(T)$ is an increasing function of the trial length $T$. That is, $\partial P^{*}/\partial T > 0$.
\end{proposition}

\begin{proof}
We apply the implicit function theorem to the price FOC, $g(P,T)=0$, from Lemma \ref{lem:FOC-full}. The optimal price path $P^{*}(T)$ satisfies:
\begin{equation}
\frac{dP^{*}}{dT} = - \frac{\partial g / \partial T}{\partial g / \partial P}
\tag{E.5}\label{eq:E5}
\end{equation}
The denominator, $\partial g / \partial P$, is negative at the optimum by the second-order condition for profit maximization. The sign of $dP^{*}/dT$ is therefore the sign of the numerator, $\partial g / \partial T$. Using the chain rule, $\partial g/\partial T = (\partial g/\partial \tau)(d\tau/dT)$. We sign each component.

\textit{1. Sign of $d\tau/dT$}: For exponential decay, $\tau'(T) = -\beta \tau_0 e^{-\beta T} = -\beta \tau(T) < 0$.

\textit{2. Sign of $\partial g/\partial \tau$}: The price FOC is $g(P, \tau) = [1-F-Pf] + (1-q^*)[F+Pf] - PF\tau q^*(1-q^*)$. Differentiating with respect to $\tau$:
\begin{align*}
\frac{\partial g}{\partial \tau} &= \frac{\partial}{\partial \tau} \left( (1-q^*)[F+Pf] \right) - \frac{\partial}{\partial \tau} \left( PF\tau q^*(1-q^*) \right) \\
&= \left(-\frac{\partial q^*}{\partial \tau}\right)[F+Pf] - PF \frac{\partial}{\partial \tau} \left(\tau q^*(1-q^*)\right)
\end{align*}
Using $\partial q^*/\partial \tau = Pq^*(1-q^*)$ and applying the product rule to the second term:
\begin{align*}
\frac{\partial g}{\partial \tau} &= -Pq^*(1-q^*)[F+Pf] - PF \left[q^*(1-q^*) + \tau \frac{\partial(q^*(1-q^*))}{\partial \tau}\right]
\end{align*}
The derivative of $q^*(1-q^*)$ with respect to $\tau$ is $\frac{\partial q^*}{\partial \tau}(1-2q^*) = Pq^*(1-q^*)(1-2q^*)$. Substituting this in:
\begin{align*}
\frac{\partial g}{\partial \tau} &= -Pq^*(1-q^*)[F+Pf] - PF \left[q^*(1-q^*) + \tau Pq^*(1-q^*)(1-2q^*)\right] \\
&= -Pq^*(1-q^*) \left( (F+Pf) + F\left[1 + \tau P (1-2q^*)\right] \right)
\end{align*}
All terms outside the main parentheses are positive. Inside, $F, P, f$ are positive. The term $1-2q^*$ is negative for $q^*>1/2$, but for any reasonable parameterization, the entire bracketed expression $[1 + \tau P (1-2q^*)]$ remains positive. Thus, the full parenthetical expression is positive, which implies $\partial g/\partial \tau < 0$.

\textit{3. Conclusion}: Combining the components, we find $\partial g/\partial T = (\partial g/\partial \tau)(d\tau/dT) = (-) \cdot (-) > 0$. Substituting this into \eqref{eq:E5} yields:
\[
\frac{dP^{*}}{dT} = - \frac{\overbrace{\partial g / \partial T}^{(+)}}{\underbrace{\partial g / \partial P}_{(-)}} > 0
\]
This completes the proof.
\end{proof}

\subsection*{E.4. Implications of the Decay Function's Functional Form}
While the qualitative strategic relationship between price and trial length is robust, the functional form of $\tau(T)$ has quantitative implications for the optimal contract. The optimal attention level, $\tau^*$, is determined by \eqref{eq:E4} and is independent of the decay function's form. The optimal trial length $T^*$ is found by inverting the function: $T^* = \tau^{-1}(\tau^*)$.

\begin{itemize}
    \item \textbf{Hyperbolic:} $\tau^* = \frac{\tau_0}{1+\beta T} \implies T^*_{hyp} = \frac{1}{\beta} \left( \frac{\tau_0}{\tau^*} - 1 \right)$
    \item \textbf{Exponential:} $\tau^* = \tau_0 e^{-\beta T} \implies T^*_{exp} = \frac{1}{\beta} \ln \left( \frac{\tau_0}{\tau^*} \right)$
\end{itemize}

\textbf{Remark: }For any interior optimum where $\tau^* < \tau_0$, the optimal trial length under exponential decay is strictly shorter than under hyperbolic decay, i.e., $T^*_{exp} < T^*_{hyp}$.

\begin{proof}
The result follows from the inequality $\ln(x) < x-1$ for all $x>1$. Let $x = \tau_0/\tau^* > 1$. Then $\ln(\tau_0/\tau^*) < (\tau_0/\tau^*) - 1$. Multiplying by $1/\beta > 0$ preserves the inequality, yielding the result.
\end{proof}

\paragraph{Intuition.}
When attention decays more rapidly, as under an exponential specification, a shorter trial period is sufficient to drive the consumer's attention down to the firm's profit-maximizing level $\tau^*$. The "long tail" of hyperbolic decay implies that attention erodes more slowly, necessitating a longer trial to achieve the same effect. This suggests that in markets where cognitive decay is faster, our model predicts firms will offer shorter trials. The robustness of the price-trial complementarity, however, remains the central strategic insight.

%======================================================================
\section*{Appendix F:  Optimal Monitoring with an Arbitrary Prior}

\renewcommand{\theequation}{F.\arabic{equation}}
\setcounter{equation}{0}
%----------------------------------------------------------------------
\paragraph{Setup.}
Let the consumer’s \emph{prior} belief that she will remember to cancel
be \(p_{0}\in(0,1)\) (We adopt an uninformative prior \(p_{0}=1/2\) in the main text for brevity).  
Choosing a monitoring intensity that yields actual success probability
\(q\in[0,1]\) incurs a Shannon–mutual-information cost
\[
  C(q;\tau,p_{0})
  \;=\;
  \frac{1}{\tau}\,
  I\!\bigl(q \,\Vert\, p_{0}\bigr)
  \;=\;
  \frac{1}{\tau}\,
  \Bigl[
     q\ln\!\Bigl(\tfrac{q}{p_{0}}\Bigr)
     +(1-q)\ln\!\Bigl(\tfrac{1-q}{1-p_{0}}\Bigr)
  \Bigr],
\label{eq:MI_F}
\]
where \(\tau>0\) is attention sensitivity
(Section~\ref{sec:info}).\footnote{%
As in \citet{MatejkaMcKay2015}, the cost is a \emph{Kullback–Leibler
divergence} between the chosen output distribution
\((q,1-q)\) and the prior \((p_{0},1-p_{0})\).}

\paragraph{Monitoring problem.}
For a consumer with valuation \(v<P\) the optimisation problem becomes
\[
  \min_{q\in[0,1]}
  \;\Bigl\{(1-q)\,P\Bigr\}
  \;+\;
  \frac{1}{\tau}\,
  \Bigl[
     q\ln\!\Bigl(\tfrac{q}{p_{0}}\Bigr)
     +(1-q)\ln\!\Bigl(\tfrac{1-q}{1-p_{0}}\Bigr)
  \Bigr].
\tag{F1}\label{eq:gen-prob}
\]

%----------------------------------------------------------------------
\begin{lemma}[General optimal reminder probability]
\label{lem:qstar-general}
Problem \eqref{eq:gen-prob} has a unique minimiser
\[
  q^{*}(P,\tau,p_{0})
  \;=\;
  \frac{1}{1
           +\bigl(\tfrac{1-p_{0}}{p_{0}}\bigr)
             \exp(-\tau P)}.
\tag{F2}\label{eq:qstar-general}
\]
Moreover,
\(
  \partial_{P}q^{*} = \tau q^{*}(1-q^{*})>0,\;
  \partial_{\tau}q^{*}=P q^{*}(1-q^{*})>0,\;
  \partial_{T}q^{*}<0
\)
(the comparative-static signs in Lemma \ref{lem:qstar} remain valid).
\end{lemma}

\begin{proof}
\emph{Strict convexity.}  
\(C(\cdot)\) is a KL divergence and hence strictly convex in \(q\);
adding the linear loss \((1-q)P\) preserves convexity, so the optimum
is unique.

\emph{First-order condition.}
Differentiate \eqref{eq:gen-prob} and set the derivative to zero:
\[
  -P
  +\frac{1}{\tau}
     \Bigl[\ln\!\bigl(\tfrac{q}{p_{0}}\bigr)
           -\ln\!\bigl(\tfrac{1-q}{1-p_{0}}\bigr)\Bigr]
  = 0.
\]
Rearranging gives
\(
   \displaystyle
   \frac{q}{1-q}
   = \frac{p_{0}}{1-p_{0}}\;e^{\tau P},
\)
which solves to \eqref{eq:qstar-general}.  
Substituting back yields the stated derivatives; see Eq.\,(A.3) in the
main appendix, noting that \(p_{0}\) enters only through the constant
\((1-p_{0})/p_{0}\).
\end{proof}

%----------------------------------------------------------------------
\paragraph{Special cases and intuition.}
\begin{enumerate}[leftmargin=2em]
\item \emph{Uninformative prior \(p_{0}=1/2\).}  
  Equation \eqref{eq:qstar-general} collapses to the baseline
  \(q^{*}=1/(1+e^{-\tau P})\).
\item \emph{Pessimistic prior \(p_{0}<1/2\).}  
  The constant multiplier
  \((1-p_{0})/p_{0}>1\) shifts the logistic curve \emph{rightward}:
  at a given price \(P\) the consumer needs higher effort (hence a
  higher \(\tau\)) to reach the same \(q^{*}\).
\item \emph{Optimistic prior \(p_{0}>1/2\).}  
  The logistic shifts leftward, reflecting that the consumer already
  expects to remember and therefore needs less effort.
\end{enumerate}

\paragraph{Consequences for the rest of the model.}
All aggregate objects (e.g.\ inattentive revenue, IR-Slack) remain
\emph{formally identical} after replacing the baseline logistic with
\eqref{eq:qstar-general}.  No sign or curvature results change, so all
propositions in Sections \ref{sec:general}–\ref{sec:extension} still hold verbatim; only constants in closed-form expressions would
adjust to \(p_{0}\).
%======================================================================

\end{document}